\documentclass[reqno]{amsart}

\usepackage{mathrsfs}
\usepackage{pb-diagram,lamsarrow,pb-lams}
\usepackage{latexsym}

\usepackage{hyperref}

\newtheoremstyle{myplain}{}{}{\it}
{0pt}{\scshape}{}{ }{\thmname{#1}\thmnumber{ #2}\thmnote{ (#3)}}
\newtheoremstyle{mydefinition}{}{}{}
{0pt}{\scshape}{}{ }{\thmname{#1}\thmnumber{ #2}\thmnote{ (#3)}}

\theoremstyle{myplain}

	\newtheorem{Def}{Definition}[section]
        \newtheorem{Lem}[Def]{Lemma}
        \newtheorem{theo}[Def]{Theorem}
        \newtheorem{prop}[Def]{Proposition}
        \newtheorem{rem}[Def]{Remark}
        \newtheorem{hyp}[Def]{Hypothesis}
        \newtheorem{cor}[Def]{Corollary}

\renewcommand{\qed}{\nopagebreak\hfill$\Box$}

\newcommand{\skp}[2]{\mbox{$\left\langle #1\, , \, #2\right\rangle$}}
\newcommand{\skpd}[2]{\mbox{$\left\langle #1\, ,\,#2\right\rangle_{\ell^2}$}}

\newcommand{\natop}[2]{\genfrac{}{}{0pt}{}{#1}{#2}}

\newcommand{\dnt}[1]{\mbox{$\frac{d #1}{dt}$}}

\DeclareMathOperator{\Span}{span}

\DeclareMathOperator{\supp}{supp}

\DeclareMathOperator{\Op}{Op}

\newcommand{\id}{\mathbf{1}}

\numberwithin{equation}{section}

\newcommand{\beqa}{\begin{eqnarray*}}
\newcommand{\eeqa}{\end{eqnarray*}}
\renewcommand{\hat}{\widehat}
\newcommand{\bauf}{\begin{itemize}}
\newcommand{\eauf}{\end{itemize}}
\newcommand{\be}{\begin{equation}}
\newcommand{\ee}{\end{equation}}
\newcommand{\ben}{\begin{enumerate}}
\newcommand{\een}{\end{enumerate}}
\newcommand{\ra}{\rightarrow}

\renewcommand{\O}{\Omega}
\newcommand{\ep}{\varepsilon}

\newcommand{\R}{{\mathbb R} }
\newcommand{\Z}{{\mathbb Z}}
\newcommand{\C}{{\mathbb C}}
\newcommand{\N}{{\mathbb N}}
\newcommand{\T}{{\mathbb T}}

\newcommand{\disk}{(\varepsilon {\mathbb Z})^d}

\newcommand{\Ce}{\mathscr C}

\newcommand{\De}{\mathscr D}

\title{Agmon-type estimates for a class of difference operators}

\author{Markus Klein \and Elke Rosenberger}

\address{Markus Klein\\ Universit\"at Potsdam\\ Institut f\"ur Mathematik \\ Am Neuen Palais 10\\ 14469 Potsdam }
\email{mklein@math.uni-potsdam.de}
\address{
Elke Rosenberger\\ Universit\"at Potsdam\\ Institut f\"ur Mathematik \\ Am Neuen Palais 10\\ 14469 Potsdam}
\email{erosen@rz.uni-potsdam.de}

\date{\today}


\begin{document}

\begin{abstract}
We analyze a general class of self-adjoint difference operators $H_\ep = T_\ep + V_\ep$ on
$\ell^2(\disk)$,
where $V_\ep$ is a one-well potential and $\ep$ is a small parameter.
We construct a Finslerian distance $d$ induced by $H_\ep$ and show that short integral curves are
geodesics.
Then we show that Dirichlet eigenfunctions decay exponentially with a rate
controlled by the Finsler distance to the well. This is analog to semiclassical
Agmon estimates for Schr\"odinger operators.
\end{abstract}

\maketitle

\section{Introduction}

The central topic of this paper is the investigation of a rather general class of families of
self-adjoint difference operators $H_\ep$ on the Hilbert space $\ell^2(\disk)$, as the small parameter  $\ep>0$
tends to zero.

The operator $H_\ep$ is given by
\begin{align} \label{Hepein}
H_\ep = (T_\ep + V_\ep ),  \quad&\text{where}\quad
T_\ep  = \sum_{\gamma\in\disk} a_\gamma \tau_\gamma ,\\
(\tau_\gamma u)(x) = u(x+\gamma) \quad &\text{and}\quad (a_\gamma u)(x) := a_\gamma(x,\ep) u(x) \quad \mbox{for} \quad x,\gamma\in\disk
\end{align}
where $V_\ep$ is a multiplication operator, which in leading order is given by
$V_0 \in \Ce^\infty (\R^d)$.

We remark that the limit $\ep \to 0$ is analog to the semiclassical limit $\hbar \to 0$
for the Schr\"odinger operator
$-\hbar^2 \Delta + V$. This paper is the first in a series of papers; the aim is to develop an analytic approach
to the semiclassical eigenvalue problem and tunneling for $H_\ep$ which is comparable  in detail and
precision
to the well known analysis for the Schr\"odinger operator (see Simon \cite{Si1}, \cite{Si2} and
Helffer-Sj\"ostrand \cite{hesjo}). Our motivation comes from
stochastic problems (see Bovier-Eckhoff-Gayrard-Klein \cite{begk1}, \cite{begk2}). A large class of discrete
Markov chains analyzed in \cite{begk2}
with probabilistic
techniques falls into the framework of difference operators treated in this article.

We recall that sharp
semiclassical Agmon estimates describing the exponential decay of eigenfunctions of appropriate Dirichlet
realizations of the Schr\"odinger operator are crucial to analyze tunneling for the Schr\"odinger operator.
We further recall that the original work of Agmon on the decay of eigenfunctions for second order
differential operators is not in the semiclassical limit. It treats the limit $|x|\to \infty$ (in a 
non bounded domain of $\R^n$). 
Agmon realized in \cite{agmon} that for a large class of such operators the exponential rate at 
which eigenfunctions decay is given by the geodesic distance
in the Agmon metric. This is the Riemannian metric from Jacobi's theorem in classical mechanics:
For a Hamilton function whose kinetic energy is a positive definite quadratic form in the momenta,
the projection to configuration
space of an integral curve of the Hamiltonian vector field is a geodesic in the Agmon (Jacobi) metric.

This paper contains analog results for the class of operators $H_\ep$, including a generalization of
Jacobi's theorem. It is essential that we consider these operators as semiclassical quantizations
of suitable Hamilton functions and investigate the relation of these Hamilton functions to Finsler geometry.
In this generality our results are new. We recall, however, that various examples extending the original framework
of the semiclassical analysis in the work of Simon \cite{Si1} and Helffer-Sj\"ostrand \cite{hesjo} have been analyzed:
The operator $\cos hD_x + \cos x$ in Harpers equation (see e.g. Helffer-Sj\"ostrand \cite{harper}), the Schr\"odinger operator with magnetic field (Helffer-Mohamed \cite{helmoh}), the Dirac and Klein-Gordon operator (see e.g.
Helffer-Parisse \cite{helpar}, Servat \cite{servat}) and the Kac operator (Helffer \cite{kac}).

If $\T^d := \R^d/(2\pi)\Z^d$ denotes the $d$-dimensional torus and
$b\in \Ce^\infty\left(\R^d\times \T^d\times (0,1]\right)$,
a pseudo-differential operator $\Op_\ep^{\T^d}(b): {\mathcal K}\left(\disk\right) \longrightarrow
{\mathcal K}'\left(\disk\right)$ is defined by
\begin{equation}\label{psdo2}
\Op_\ep^{\T^d}(b)\, v(x) := (2\pi)^{-d} \sum_{y\in\disk}\int_{[-\pi,\pi]^d} e^{\frac{i}{\ep}(y-x)\xi}
b(x,\xi;\ep)v(y) \, d\xi \, ,
\end{equation}
where
\begin{equation}\label{kompaktge}
{\mathcal K}\left(\disk\right):=\{ u: \disk\rightarrow \C\; |\; u~\mbox{has compact
support}\}
\end{equation}
and ${\mathcal K}'\left(\disk\right):= \{f: \disk\rightarrow \C\ \} $ is dual to ${\mathcal K}\left(\disk\right)$
by use of the scalar product $\skpd{u}{v}:= \sum_x \bar{u}(x)v(x)$.

We remark that under certain assumptions on the $a_\gamma$ defining $T_\ep$ in 
\eqref{Hepein}, one has $T_\ep = \Op_\ep^{\T^d}(t(.,.;\ep))$, where 
$t\in\Ce^\infty\left(\R^d\times\T^d\times (0,1]\right)$ is given by
\begin{equation}\label{talsexp}
t(x,\xi, \ep) = \sum_{\gamma\in\disk} a_\gamma (x,\ep) \exp \left(-\frac{i}{\ep}\gamma\cdot\xi\right)\; .
\end{equation}
Here $t(x,\xi;\ep)$ is considered as a function on $\R^{2d}\times (0,1]$, which is
$2\pi$-periodic with respect to $\xi$. 

Furthermore, assuming that $a_\gamma(x,\ep) = a^{(0)}_\gamma(x) + \ep
a^{(1)}_\gamma(x) + R^{(2)}_\gamma(x,\ep)$, where
$R^{(2)}_\gamma(x,\ep)= O(\ep^2)$ uniformly with respect to $x$ and $\gamma$, we can write
\begin{align}\label{texpand}
t(x,\xi;\ep) &= t_0 (x,\xi) + \ep\,  t_1(x,\xi) + t_2(x,\xi;\ep)\; ,\qquad\text{with}\\
t_j(x,\xi) &:= \sum_{\gamma\in\disk} a_\gamma^{(j)}(x) e^{-\frac{i}{\ep}\gamma\xi}\, , \qquad j=0,1\nonumber\\
t_2(x,\xi;\ep) &:= \sum_{\gamma\in\disk} R_\gamma^{(2)}(x,\ep) e^{-\frac{i}{\ep}\gamma\xi}\nonumber\; .
\end{align}

Thus, in leading order the symbol of $H_\ep$ is $h_0=t_0+V_0$.
In its original form, neither Jacobi's theorem applies to $h_0(x, \xi)$ nor Agmon estimates to $H_\ep$.
Our analysis is motivated by the remark in Agmon's book \cite{agmon} to develop part of the theory of
the Agmon metric in
the more general context of Finsler geometry. It turns out that the Hamilton function
$-h_0(x,i\xi)$  (this transformation is analog to the procedure in the case
of the Schr\"odinger operator) in a natural way induces a Finsler metric and an associated
Finsler distance $d$ on $\R^d$. This allows to formulate and prove a generalization of
Jacobi's theorem (which might be some kind of lesser known folk wisdom in mathematical physics, which, however,
we were unable to find in the literature) and prove an analog of the semiclassical Agmon estimates for $H_\ep$.
We remark that Finsler distances have been used for higher order elliptic differential
operators in the analysis of decay of resolvent kernels and/or heat kernels,
see Tintarev \cite{tinta} and Barbatis \cite{barbatis1}, \cite{barbatis2}.\footnote{M.K.  thanks S. Agmon for 
the reference to \cite{tinta}, where prior to the publication of Agmon's book
a Finslerian approach was used to obtain estimates on the kernel of the resolvent and the decay of the heat kernel
for higher order elliptic operators, following ideas of Agmon.}
However, these papers do not develop a generalization of Jacobi's theorem, which turns out to be crucial in our semiclassical
analysis.

We will now state our assumptions on $H_\ep$ and formulate our results more precisely.

\begin{hyp}\label{hypdecay}
\begin{enumerate}
\item The coefficients $a_\gamma(x, \ep)$ in \eqref{Hepein} are
functions
\begin{equation}\label{agammafunk}
a: \disk \times \R^d \times (0,1] \ra \R\, , \qquad (\gamma, x,
\ep) \mapsto a_\gamma(x,\ep)\, ,
\end{equation}
satisfying the following conditions:
\ben
\item[(i)] They have an
expansion
\begin{equation}\label{agammaexp}
a_\gamma(x,\ep) = a_\gamma^{(0)}(x) + \ep \, a_\gamma^{(1)}(x) +
R^{(2)}_\gamma (x, \ep)\, ,
\end{equation}
where $a_\gamma^{(i)}\in\Ce^\infty(\R^d)$ and $|a_\gamma^{(j)}(x)
- a_\gamma^{(j)}(x+ h)| = O(|h|)$ for $j=0,1$ uniformly with respect to $\gamma\in\disk$ and $x\in\R^d$.
Furthermore $R^{(2)}_\gamma \in\Ce^\infty(\R^d\times (0,1])$ for
all $\gamma\in\disk$.
\item[(ii)] $\sum_\gamma a_\gamma^{(0)}  = 0$ and $a_\gamma^{(0)}
\leq 0$ for $\gamma \neq 0$
\item[(iii)] $a_\gamma(x, \ep) =
a_{-\gamma}(x+\gamma, \ep)$ for $x \in \R^d, \gamma \in \disk$
\item[(iv)] For any $c>0$ there exists $C>0$ such that for $j=0,1$ uniformly
with respect to $x\in\disk$ and $\ep$
\begin{equation}\label{abfallagamma}
\| \, e^{\frac{c|.|}{\ep}} a^{(j)}_.(x)\|_{\ell_\gamma^2(\disk)}\leq C \qquad\text{and} \qquad
\|\,
 e^{\frac{c|.|}{\ep}} R^{(2)}_.(x)\|_{\ell^2_\gamma(\disk)}
 \leq C\ep^2
\end{equation}
\item[(v)]
$\Span \{\gamma\in\disk\,|\, a^{(0)}_\gamma(x) <0\}= \R^d$ for all $x\in\R^d$.
\een
\item 
\ben
\item[(i)] The potential energy $V_\ep$ is the restriction to $\disk$ of a
function $\hat{V}_\ep\in\Ce^\infty (\R^d, \R)$, which has an expansion
\[
\hat{V}_{\ep}(x) = \sum_{l=0}^N\ep^l V_l(x) + R_{N+1}(x;\ep)  \, ,
\]
where $V_\ell\in\Ce^\infty(\R^d)$, $R_{N+1}\in \Ce^\infty (\R^d\times (0,\ep_0])$ for some $\ep_0>0$ and
for any compact set $K\subset \R^d$
there exists a constant $C_K$ such that $\sup_{x\in K} |R_{N+1}(x;\ep)|\leq C_K \ep^{N+1}$.
\item[(ii)]
There exist constants $R, C > 0$ such that
$V_\ep(x) > C$ for all $|x| \geq R$ and $\ep\in(0,\ep_0]$.
\item[(iii)]
$V_0(x)$ has exactly one non-degenerate minimum at $x_0=0$ with the value
$V_{0}(0)=0$.
\een
\een
\end{hyp}

The following lemma couples the assumptions on the coefficients $a_\gamma$ given in Hypothesis \ref{hypdecay}
with properties of the symbol $t$ and the kinetic energy $T_\ep$.

\begin{Lem}\label{propt}
Assume Hypothesis \ref{hypdecay} and let $t$ and $t_j, j=0,1,2$ be defined in \eqref{talsexp} and \eqref{texpand}
respectively. Then:
\ben
\item $t\in \Ce^\infty(\R^d\times \T^d\times (0,1])$ and the estimate
$\sup_{x,\xi}|\partial_x^\alpha\partial_\xi^\beta
t(x,\xi;\ep)|\leq C_{\alpha,\beta}$ holds for all $\alpha,\beta\in\N^d$ uniformly with respect to $\ep$.
Furthermore $t_0$ and $t_1$ are bounded and $\sup_{x,\xi}|t_2(x,\xi;\ep)| = O(\ep^2)$.
\item The $2\pi$-periodic function $\R^d\ni\xi\mapsto t_0(x,\xi)$ is even and has an analytic continuation to $\C^d$.
\item At $\xi=0$, for fixed $x\in\R^d$ the function $t_0$ defined in \eqref{texpand} has an expansion
\begin{equation}\label{kinen}
t_0(x,\xi) = \skp{\xi}{B(x)\xi} + O\left(|\xi|^4\right)\qquad\text{as}\;\; |\xi|\to 0\, ,
\end{equation}
where $B:\R^d\rightarrow\mathcal{M}(d\times d,\R)$ is positive
definite and symmetric.
\item The operator $T_\ep$ defined in \eqref{Hepein} is symmetric, bounded (uniformly in $\ep$) and
$\skpd{u}{T_\ep u} \geq -C\ep \|u\|^2$ for some $C>0$.
Furthermore $T_\ep=\Op_\ep^{\T^d}(t)$ (see \eqref{psdo2}).
\een
\end{Lem}

\begin{rem}
$T_\ep$ being symmetric  boils down to $a_\gamma$ being real and condition (a)(iii) of Hypothesis \ref{hypdecay}. In
the probabilistic context, which is our main motivation, the latter is a standard reversibility condition while the
former ist automatic for a Markov chain (see Section \ref{ananwend}). Since $T_\ep$ is bounded, $H_\ep=T_\ep + V_\ep$ defined in
\eqref{Hepein} posseses a self adjoint realization on the maximal domain of $V_\ep$. Abusing notation, we
shall denote this realization also by $H_\ep$ and its domain by $\De(H_\ep)\subset \ell^2\left(\disk\right)$. The associated symbol is denoted by $h(x,\xi;\ep)$. Clearly, $H_\ep$ commutes with complex conjugation.
\end{rem}

We will use the notation
\begin{equation}\label{agammaunep}
\tilde{a}:\Z^d\times\R^d\ni (\eta,x)\mapsto \tilde{a}_\eta(x) := a_{\ep\eta}^{(0)}(x)\in\R  
\end{equation}
and set
\begin{equation}\label{tildehnull}
\tilde{h}_0 (x,\xi):= -h_0(x,i\xi) = {\tilde t}_0(x,\xi) - V_0(x)\,:\, \R^{2d} \ra \R\; ,
\end{equation}
where by Lemma \ref{propt} (b)
\begin{equation}\label{tildetdef}
\tilde{t}_0(x,\xi):= -t_0(x,i\xi) = -\sum_{\eta\in\Z^d} \tilde{a}_\eta(x) \cosh \left(\eta\cdot \xi\right) \; .
\end{equation}

We shall now describe, how Hamilton functions such as $\tilde{h}_0$ for fixed energy $E$ introduce a Finsler
geometry in configuration space.

\begin{hyp}\label{hypfins}
Let $M$ be a $d$-dimensional smooth manifold.
Let $h\in\Ce^\infty(T^*M,\R)$ be hyperregular and even and strictly convex in
each fibre. Furthermore, let $h(.,0)$ be bounded from above.
For $E\in\R$ set $\widetilde{M}:= M\setminus \{h(x,0)\geq E\}$.
Denoting the fibre derivative of $h$ by ${\mathcal D}_Fh$, we
associate to $h$ the energy function
$E_h (x,v) :=h \circ \left({\mathcal D}_F h\right)^{-1} (x,v)$ on $TM$.
\end{hyp}
The notion of fibre derivative and hyperregular are standard (see
Abraham-Marsden \cite{abma}). For convenience of the reader, they are repeated in Definition \ref{fibreder}.

Now Theorem \ref{listfinsler} states that assuming Hypothesis \ref{hypfins}
\begin{equation}\label{leins}
\ell_{h,E}(x,v) := \left({\mathcal D}_Fh\right)^{-1}(x,\tilde{v})\cdot v\; ,
\end{equation}
where $\tilde{v}$ is chosen such that $E_h(x,\tilde{v}) = E$, is a Finsler function on $\widetilde{M}$.
The most important property of $\ell_{h,E}$ is the homogeneity
\[ \ell_{h,E}(x,\lambda v) = |\lambda| \ell_{h,E}(x,v)\, , \qquad \lambda\in\R\; ,\]
which is analog to the homogeneity of $|v|=\sqrt{g(v,v)}$ in the case of a Riemannian metric.
This is essential to define a curve length associated to $\ell_{h,E}$
as described in Definition \ref{Finslerman2} by
\[ s_{\ell_{h,E}}(\gamma) := \int_{a}^{b} \ell_{h,E}(\gamma(t),\dot{\gamma} (t)) \, dt \, . \]
A Finsler geodesic is then a curve $\gamma$ on $M$, for which $s_{\ell_{h,E}}$ is extremal
(see Def. \ref{Finslerman3a}).

The following theorem establishes the connection between geodesics with respect
to the Finsler function $\ell_{h,E}$ for a given hyperregular Hamilton function $h$
and the integral curves of the associated Hamiltonian vector field $X_h$.
It amplifies the Maupertuis principle in classical mechanics.

\begin{theo}\label{eulermaup}
Let $h$, $E$ and $\widetilde{M}$ satisfy Hypothesis \ref{hypfins}. Let $\ell_h :=\ell_{h,E}$
be as defined in \eqref{leins} (see Theorem \ref{listfinsler} for details).
\ben
\item
Let $\gamma_0:[a,b]\ra \widetilde{M}$ be a base integral curve of the Hamiltonian vector field $X_h$
with energy $E$ (i.e. $E_h(\gamma_0(t),\dot{\gamma}_0(t)) = E$ for all $t\in[a,b]$).
Then $\gamma_0$ is a geodesic on $\widetilde{M}$ with respect to $\ell_h$.
\item Conversely, if $\gamma_0$ is a geodesic on $\widetilde{M}$ with respect to $\ell_h$ with energy
$E_h(\gamma_0,\dot{\gamma}_0) = E$, then $\gamma_0$ is a base integral curve of $X_h$.
\een
\end{theo}

The Hamilton function $\tilde{h}_0$ introduced
in \eqref{tildehnull} actually satisfies Hypothesis \ref{hypfins} with respect to the energy $E=0$ (see Corollary
\ref{hnullhyperregular})
and thus induces a Finsler function $\ell:=\ell_{\tilde{h}_0,0}$ and a
Finsler distance defined by
\begin{equation}\label{delleins}
d_\ell(x_0,x_1) = \inf_{\gamma\in\Gamma_{0,1}(x_0, x_1)} \int_0^{1}
\ell(\gamma(t), \dot{\gamma}(t))\, dt
\, ,
\end{equation}
where $\Gamma_{0,1}(x_0,x_1)$ denotes the set of regular curves $\gamma$ with $\gamma(0) = x_0$ and $\gamma(1) = x_1$.

\begin{theo}\label{eikonald}
There exists a neighborhood $\Omega$ of $0$ such that
$d^0(x) := d_\ell(0,x) $, with $d_{\ell}$ defined in \eqref{delleins}, fulfills the generalized eikonal equation
\[ \tilde{h}_0(x,\nabla d^0 (x)) = 0\, , \qquad x\in\Omega \; . \]
Furthermore
\begin{equation}\label{eicmitd}
 d^0(x) - \sum_{1\leq k \leq N} \varphi_k(x) = O(|x|^{N+1}) \quad \mbox{as} \quad x\to 0\; ,
\end{equation}
where
each $\varphi_k$ is an homogeneous
polynomial of degree $k+2$.

In addition $d_\ell$ is locally Lipschitz continuous, i.e.
\begin{equation}\label{d-dgamma}
|d_\ell(x,y)|  \leq C\, |x-y| \; ,\qquad x,y\in\R^d\, ,
\end{equation}
where $C$ is locally uniform in $x$ and $y$.

The eikonal inequality
\begin{equation}\label{eicungl}
\tilde{h}_0 (x,\nabla d^0(x)) \leq 0
\end{equation}
holds almost everywhere in $\R^d$.
\end{theo}

To analyze eigenfunctions concentrated at the potential minimum $x_0=0$, we introduce
a Dirichlet operator $H_\ep^\Sigma$ as follows.

\begin{Def}\label{ell2sigma}
For $\Sigma\subset \R^d$ we set $\Sigma_\ep :=\Sigma \cap \disk$.
Any function $u\in\ell^2(\Sigma_\ep)$ can by zero extension, i.e. via
$u(x)=0$ for $x\notin \Sigma_\ep$, be embedded in
$\ell^2(\disk)$.
If we denote this embedding by $i_{\Sigma_\ep}$, we can define the space
$\ell^2_{\Sigma_\ep}:= i_{\Sigma_\ep} \left(\ell^2(\Sigma_\ep)\right) \subset \ell^2(\disk)$ and the Dirichlet operator
\begin{equation} \label{HepD}
H_\ep^\Sigma :=\id_{\Sigma_\ep} H_\ep|_{\ell^2_{\Sigma_\ep}}  \;:\; \ell^2_{\Sigma_\ep} \rightarrow \ell^2_{\Sigma_\ep}
\end{equation}
with domain $\De (H_\ep^\Sigma) = \{u\in\ell^2_{\Sigma_\ep}\,|\, V_\ep u \in \ell^2_{\Sigma_\ep}\}$.
\end{Def}

We now formulate our estimates of weighted $\ell^2$-norms of eigenfunctions of the Dirichlet operator
$H^\Sigma_\ep$.
We will show that they decay exponentially
at a rate controlled by the Finsler distance $d^0(x)$.
Theorem \ref{eikonald} is crucial to prove these estimates.

\begin{theo}\label{weig}
Let $\Sigma\subset\R^d$ be a bounded open region including the point $0$ such that
$d^0\in{\mathscr C}^2(\overline{\Sigma})$, where
$d^0(x):=d_\ell(0,x)$ is defined by
\eqref{delleins}.

Let $E\in [0,\ep R_0]$ for $R_0$ fixed,
assume Hypothesis \ref{hypdecay} and let $H_\ep^{\Sigma}$ denote the
Dirichlet operator introduced in (\ref{HepD}).

Then there exist constants $\ep_0, B, C>0$ such that for all $\ep\in(0,\ep_0]$
and real $u\in \ell^2_{\Sigma_\ep}$
\begin{equation}\label{weigequ}
\left\| \left(1+\tfrac{d^0}{\ep}\right)^{-B} e^{\frac{d^0}{\ep}} u
\right\|_{\ell^2}
\leq
C \left[ \ep^{-1}\left\| \left(1+\tfrac{d^0}{\ep}\right)^{-B}
e^{\frac{d^0}{\ep}}
\left(H_\ep^{\Sigma}-E\right)u\right\|_{\ell^2} +
 \| u \|_{\ell^2}  \right]
\; .
\end{equation}
In particular, let
$u\in \ell_{\Sigma_\ep}^2$ be a normalized eigenfunction of $H_\ep^\Sigma$
with respect to the eigenvalue
$E\in[0,\ep R_0]$. Then there exist constants $B,C>0$ such
that for all $\ep\in(0,\ep_0]$
\begin{equation}\label{eigenu}
 \left\| \left(1+\frac{d^0}{\ep}\right)^{-B} e^{\frac{d^0}{\ep}}u\right\|_{\ell^2} \leq
 C \, .
 \end{equation}
\end{theo}

\begin{rem} The estimate \eqref{weigequ} ist the analog of the sharp semiclassical Agmon estimate in 
Helffer-Sj\"ostrand \cite{hesjo} for the Schr\"odinger operator. We emphasize that our generalization
in the context of general Finsler geometry is a result in the semiclassical limit, under the 
crucial hypothesis that both the kinetic and the potential energy have a non-degenerate minimum at
$\xi=0, x=0$. 
In particular, we have nothing to report for an analog of the original Agmon estimate if $\ep=1$ and
$|x|\to\infty$. 
\end{rem}

The plan of the paper is as follows.

Section \ref{Kap5} is devoted to the construction and properties of a Finsler function associated to a hyperregular
Hamilton function. In particular, in Subsection \ref{fins1} we introduce the general notion of a Finsler manifold,
the associated curve length and Finsler geodesics. In Subsection \ref{Finshyp} we construct the absolute
homogeneous Finsler function $\ell_{h,E}$ with respect to an hyperregular Hamilton function $h$ and a
fixed energy $E$. In particular, we prove Theorem \ref{listfinsler}. The proof of Theorem \ref{eulermaup}
is given in Subsection \ref{geodesic}.
In Subsection \ref{Finseik} we prove Lemma \ref{propt} and we show that we can apply the results derived up to this point to the
Hamilton function $\tilde{h}_0$ defined in \eqref{tildehnull}.
Subsection \ref{proofeiko} contains the proof of Theorem \ref{eikonald}.

In Section \ref{Kap5a} we
show the exponential decay of the eigenfunctions of the low lying spectrum of $H_\ep$
with a rate controlled by the Finsler distance constructed in Section \ref{Kap5}.
In particular, in  Subsection \ref{pre} we show three basic lemmata and in Subsection \ref{proofweig} we
prove Theorem \ref{weig}.
In Section \ref{ananwend} we describe how a certain class of Markov chains fits into the framework of our
hypotheses.

\section{Finsler Distance associated to $H_\ep$}\label{Kap5}

\subsection{Definition and Properties of Finsler Manifold and Finsler Metric}\label{fins1}

We introduce the general notion of a Finsler manifold and Finsler
distance (for detailed
description of Finsler manifolds we refer e.g. to Bao-Chern-Shen \cite{bao},
Abate-Patrizio \cite{abate}).

For a manifold $M$, $\pi: TM\ra M$ denotes the tangent bundle with fibre $T_xM=\pi^{-1}(x)$.
We denote an element of $TM$ by $(x,v)$ where $x\in M$ and $v\in T_xM$.
Analogously, $(x,\xi)$ with $\xi \in T^*_xM$ denotes a point in the cotangent bundle $\pi^*: T^*M\ra M$.

The canonical pairing between an element $v\in T_xM$ and $\xi\in T^*_xM$ is written as $v\cdot \xi$.

\begin{Def}\label{Finslerman}
Let $M$ be a $d$-dimensional $\Ce^\infty$-manifold and
$TM\setminus \{0\} :=\{(x,v)\in TM\,|\,v\neq 0\}$ the slit tangent bundle.
\ben
\item A (Lagrange)-function $F:TM \ra [0,\infty)$ is called a Finsler function on $M$, if:
\ben
\item[1)] $F$ is of class ${\mathscr C}^\infty (TM\setminus \{0\})$.
\item[2)] $F(x,\lambda v) = \lambda F(x,v)$
for $\lambda > 0$, i.e., $F$ is positive homogeneous of order 1 in each fibre.
\item[3)] $g(x,v):= D^2_v\left(\frac{1}{2}F^2\right)|_{(x,v)}$ is positive definite as a bilinear
form on $T_xM$ for all $(x,v)\in TM\setminus \{0\}$.
\een
\item A Finsler function $F$ is said to be absolute homogeneous,
if
\ben
\item[4)] $F(x,\lambda v) = |\lambda| F(x,v)$ for all
$\lambda\in\R$,
\een
\item A manifold together with a Finsler function, $(M,F)$, is called a Finsler manifold.
\een
\end{Def}

\begin{rem}
\ben
\item By Euler's Theorem, 2) implies $\sum_{i,j=1}^d g_{ij}(x,v) v_i v_j = F^2(x,v)$ and thus 3) implies 
$F(x,v)>0$ for $v\neq 0$ (for details see Bao-Chern-Shen \cite{bao}).
\item In Agmon \cite{agmon}, the definition of a Finsler function is slightly more general: $F$ is only
required to be continuous and $F(x,.)$ to be positive and convex. In this paper, we use the
definition of Bao-Chern-Shen \cite{bao}, since the natural class of Finsler functions which we construct
turn out to be Finsler in this more narrow sense. Thus we can use results of Bao-Chern-Shen \cite{bao}.
\een
\end{rem} 

A Finsler function induces a curve length
on $M$ as follows.
A curve $\gamma: [a,b]\rightarrow M$ on M
is called regular, if it is $\Ce^2$ and $\dot{\gamma}(t)\neq 0$ for all $t\in[a,b]$.
We introduce the Banach manifold
\begin{equation}\label{gammaab}
\Gamma_{a,b}(x_1,x_2) := \{ \gamma\in\Ce^2([a,b],M)\, |\, \gamma\;\;\text{is regular and}\;\; \gamma(a) = x_1,\, \gamma(b) =
x_2\, \} \; .
\end{equation}

\begin{Def}\label{Finslerman2}
For any Finsler
function $F$ on $M$, the curve length $s_F : \Gamma_{a,b}(x_1,x_2) \ra \R$ associated
to $F$ is defined as
\[ s_F(\gamma) := \int_{a}^{b} F(\gamma(t),\dot{\gamma} (t)) \, dt \, . \]
\end{Def}

For any $\delta >0$, a regular variation of $\gamma\in \Gamma_{a,b}(x_1,x_2)$ is a $\Ce^2$-map
$\gamma_\delta:[a,b]\times (-\delta, \delta)\ra M$,
such that $\gamma_\delta (t,0) = \gamma(t)$ for all $t\in[a,b]$ and $\gamma_\delta(.,u)$ is regular for each $u\in
(-\delta,\delta)$.

Each $\Ce^2$-map $\gamma_\delta: (-\delta,\delta) \ra \Gamma_{a,b}(x_1,x_2)$ with $\gamma_\delta(0) = \gamma$
can be considered as a regular variation of $\gamma$ with fixed endpoints
(i.e. with $\gamma_\delta(a,u)=x_1$ and $\gamma_\delta(b,u) = x_2$ for all $u\in (-\delta,\delta)$).
Therefore the tangent space of $\Gamma_{a,b}(x_1,x_2)$ at a point $\eta$ is given by
\begin{equation}\label{tangentspace}
   T_\eta\Gamma_{a,b}(x_1,x_2) = \{ \partial_u \eta_\delta|_{u=0}\,|\, \eta_\delta\;\;\text{is a regular
variation of}\;\;\eta\;\;\text{with fixed endpoints}\} \; ,
\end{equation}
where $\partial_u \eta_\delta|_{u=0}$
is considered as a vector field along $\eta$, i.e.,
as a function $\partial_u \eta_\delta|_{u=0}:[a,b]\ra TM$ such that
$\partial_u \eta_\delta|_{u=0}(t)\in T_{\eta(t)}M$. Since the variation $\eta_\delta$ has
fixed endpoints, it follows that $\partial_u \eta_\delta|_{u=0}(a) = \partial_u \eta_\delta|_{u=0}(b) =0$.

\begin{Def}\label{Finslerman3a}
$\gamma\in \Gamma_{a,b}(x_1,x_2)$ is called a geodesic with respect to the Finsler
function $F$ (or a Finsler geodesic), if $ds_F|_\gamma = 0$.
\end{Def}

\begin{Def}\label{Finslerman3}
Let $(M,F)$ denote a Finsler manifold.
\ben
\item
The Finsler distance $d_F(x_1,x_2): M\times
M\rightarrow [0,\infty]$ between the points $x_1$ and $x_2$ is defined by
\[ d_F(x_1,x_2) := \inf_{\gamma\in\Gamma_{0,1}(x_1,x_2)} s_F(\gamma) \; . \]
If $\Gamma_{0,1}(x_1,x_2)$ is empty, the distance is defined to be
infinity.
\item A geodesic $\gamma$ between two points $x_1$ and $x_2$ is called
minimal, if $s_F(\gamma) = d(x_1,x_2)$.
\een
\end{Def}

It follows easily from the definitions of a Finsler function $F$ and the associated Finsler
distance $d_F$ that
$d_F(x_1,x_2) \geq 0$, where equality holds if and only if $x_1=x_2$. Furthermore 
the triangle
inequality $d_F(x_1,x_3) \leq d_F(x_1,x_2) + d_F(x_2,x_3)$ holds.
If in addition the Finsler function $F$ is absolute homogeneous, then
$d_F(x_1,x_2) = d_F(x_2,x_1)$.
Thus for an absolute homogeneous Finsler function, $(M,d_F)$ is a metric space.

\begin{Def}\label{SMPTM}
We denote by $SM:=TM/\sim_S$ the sphere bundle, where
\[ (x,v)\sim_S (y,w)\, ,\quad\text{if}\quad x=y\quad\text{and}
\quad v=\lambda w\quad \text{for any}\quad \lambda >0\, .\]
Let $\pi_s:TM\ra SM$ denote the projection $\pi_s(x,v) =[x,v]$.
\end{Def}

\subsection{The Finsler Function of a hyperregular Hamilton function}\label{Finshyp}

To define a Finsler distance for which an analog of Jacobi's Theorem holds, we briefly introduce the notion of
fibre derivatives, hyperconvexity and hyperregularity of $h$.
It is shown in Proposition \ref{hyperconvex} that hyperconvexity of $h$ is a sufficient condition for
hyperregularity.\\

\begin{Def}\label{fibreder}
\ben
\item
Let $M$ be a manifold and $f\in{\mathscr C}^\infty\left(T^*M,\R\right)$. Then for $f_x:= f|_{T^*_xM}$ the map
${\mathcal D}_Ff:T^*M \ra TM$ defined by
${\mathcal D}_F f(x,\xi) := Df_x(\xi)$ is called the fibre derivative of $f$.
\item Analogously, the fibre derivative of a function $g\in{\mathscr C}^\infty\left(TM,\R\right)$
is defined as ${\mathcal D}_Fg:TM \ra T^*M\, ,\quad {\mathcal D}_F g(x,v) := D g_x(v)$.
\item A function $f: SM \ra TM$ is called strictly fibre preserving, if $f([x,u]) \in [x,u]$, where
$[x,u]$ denotes the equivalence class with respect to $\sim_S$.
\item A smooth function $h:T^*M \ra \R$ (or $L:TM \ra \R$) is said to be hyperregular, if its
fibre derivative ${\mathcal D}_Fh:T^*M\ra TM$ (or ${\mathcal D}_FL:TM\ra T^*M$) is a diffeomorphism.
For a hyperregular function $h\in\Ce^\infty(T^*M)$, we sometimes use the notation
\begin{equation}\label{xivonv}
\xi_h(x,v) := \left(\mathcal{D}_F h\right)^{-1}(x,v) \quad\text{and}\quad v_h(x,\xi) := \mathcal{D}_F h (x,\xi) \; .
\end{equation}
\een
\end{Def}

\begin{Def}\label{konvexdef}
For a normed vector space $V$ we call a function $L\in\Ce^2(V,\R)$
hyperconvex, if there exists a constant $\alpha > 0$ such that
\[ D^2L|_{v_0}(v,v)\geq\alpha \|v\|^2\quad\text{for all}\quad v_0,v\in V\, .\]
\end{Def}

We recall that a strictly convex function $L\in\Ce^2(V,\R)$ has the properties
\begin{align}\label{energie10}
L(v_1) - L(v_2) &\geq DL(v_2)(v_1-v_2)\\
\label{energie5}
D^2 L|_{v_0}[v,v]  &> 0 \\
\label{energie8}
\left(DL(v_1) - DL(v_2)\right)(v_1-v_2) &> 0
\end{align}

\begin{prop}\label{hyperconvex}
If a real valued function $h\in {\mathscr C}^\infty\left(T^*M\right)$ is hyperconvex in each fibre
$T_x^*M$, it is hyperregular.
\end{prop}

\begin{proof}

By definition, ${\mathcal D}_Fh$ is fibre preserving, thus in the coordinates on $TM$ and $T^*M$
induced from local coordinates on $M$ at $x_0$,
its derivative $D{\mathcal D}_Fh|_{(x_0,\xi_0)}$ is given by the $2d\times 2d$-matrix
\begin{equation}\label{hypercon}
 \left(\begin{array}{cc} \id & 0 \\ * & M \end{array}\right) \; ,
\end{equation}
where $M$ is the matrix representation of $D^2_\xi h|_{(x_0,\xi_0)} = D\left({\mathcal D}_F h|_{T_x^*M}\right)$.

Since $h$ was assumed to be
hyperconvex in each fibre, $M$ is positive definite and thus it follows from \eqref{hypercon} that
${\mathcal D}_Fh$ is a local diffeomorphism.

We claim that $Dh_x:T_x^*M\ra T_xM$ is bijective for all $x\in M$.
Since ${\mathcal D}_F h$ is fibre preserving, this shows that ${\mathcal D}_Fh$ is a global diffeomorphism
and finishes the proof.
Thus we fix any $x\in M$ and analyze $Dh_x$.

Since $h_x$ is strictly convex for each $x\in M$, by \eqref{energie8}
\[ (\xi - \eta) \left(Dh_x(\xi) - Dh_x(\eta)\right)>0\; ,\qquad \xi, \eta\in T^*_xM,\,
\eta\neq \xi \]
holds and thus ${\mathcal D}_F h(x,.)= Dh_x$ is injective.

To show the surjectivity, we claim that, for any $v_0\in T_xM$, the initial value problem
\begin{equation}\label{AWPxi}
v_0 = \frac{d}{dt} D h_x(\xi(t)) =  D^2 h_x(\xi(t))\cdot \dot{\xi}(t)\; , \qquad \xi(0)=0\; .
\end{equation}
has a solution $\xi(t)$ for all $t\in[0,1]$. Then
\[ v_0 = \int_0^1 D^2 h_x(\xi(t))\cdot \dot{\xi}(t)\, dt =
D h_x(\xi(1)) - D h_x(\xi(0)) = D h_x(\xi(1)) \]
and thus $D h_x$ is surjective.

Since $h$ is hyperconvex, the inverse $\left(D^2 h_x|_{\xi(t)}\right)^{-1}$ exists, thus \eqref{AWPxi} can
be rewritten as
\begin{equation}\label{AWPxi2}
\dot{\xi}(t) =  \left(D^2 h_x|_{\xi(t)}\right)^{-1}\cdot v_0\,, \qquad
\xi (0) = 0 \, .
\end{equation}
Thus \eqref{AWPxi2} is of the form $\dot{\xi}= F(\xi)$, where $F$ is locally Lipschitz.
Therefore for any $v_0\in T_xM$,
\eqref{AWPxi2}
has a solution, which
either exists for all $t\geq 0$ or
becomes infinite for a finite value of $t$.

In order to exclude that the curve $\xi$ reaches infinity for some $t<1$,
we need the hyperconvexity of $h$.
We choose a norm $\|.\|_{T_x^*M}$ on $T_x^*M$ and denote by $\|.\|_{T_xM}$ the norm on
$T_xM$, which is induced by duality. Since for fixed $\eta\in T_x^*M =T_\xi(T_x^*M)$ the second derivative
$D^2h_x|_\xi(\eta)$ can be regarded as an element of $T_xM$, it follows by the
hyperconvexity of $h$ that there exists a constant $\alpha >0$ such that for all $\xi\in T_x^*M$
\begin{equation}\label{hypercon1}
\|D^2h_x|_\xi(\eta)\|_{T_xM} = \sup_{\mu\in T_x^*M} \frac{|D^2h_x|_\xi(\eta,\mu)|}{\|\mu\|_{T_x^*M}}
\geq \frac{|D^2h_x|_\xi(\eta,\eta)|}{\|\eta\|_{T_x^*M}}\geq  \alpha \|\eta\|_{T_x^*M}\, , \quad \eta\in T_\xi(T_x^*M)
\end{equation}
and therefore
\begin{equation}\label{hypercon2}
 \|v\|_{T_xM} = \|D^2h_x|_\xi\left(D^2h_x|_\xi\right)^{-1} (v)\|_{T_xM} \geq
 \alpha \|\left(D^2h_x|_\xi\right)^{-1} (v)\|_{T^*_xM} \; ,\qquad v\in T_xM \; .
\end{equation}
\eqref{AWPxi2} together with \eqref{hypercon2} yields
\begin{equation}
 \|\dot{\xi}(t)\|_{T_x^*M} = \| \left(D^2 h_x|_{\xi(t)}\right)^{-1}(v_0)\|_{T_x^*M} \leq \frac{1}{\alpha} \|v_0\|\; .
\end{equation}
Therefore the curve $\xi(t)$ exists for all $t\in[0,1]$ and ${\mathcal D}_Fh(\xi (1))=v_0$.

\end{proof}

For any hyperregular Hamilton function $h\in {\mathscr C}^\infty\left(T^*M\right)$,
we define the energy function $E_h$ on $TM$ by
\begin{equation}\label{energiefunktion}
E_h (x,v) :=h \circ \left({\mathcal D}_F h\right)^{-1} (x,v) (= h(x,\xi_h(x,v))
\end{equation}
and the action
\begin{equation}\label{kanpaarung}
A_h: TM \ra \R\, , \quad A_h(x,v):= \left({\mathcal D}_F h\right)^{-1}(x,v)\,\cdot\, v (= \xi_h(x,v)\cdot v)\; .
\end{equation}
Then the Lagrange function
\begin{equation}\label{lagrangefunktion}
L_h:TM \ra \R \quad\text{defined by} \quad L_h(x,v) = A_h(x,v) - E_h(x,v)
\end{equation}
(the Legendre transform of h) is hyperregular on $TM$ and
\begin{equation}\label{lagrangehamilton}
\mathcal{D}_FL_h (x,v) = \left({\mathcal D}_Fh\right)^{-1}(x,v)
\end{equation}
(by Theorem 3.6.9 in \cite{abma}, the hyperregular Lagrange functions on $TM$ and the
hyperregular Hamilton functions on $T^*M$ are in bijection).
In particular, by \eqref{kanpaarung} and \eqref{lagrangehamilton},
\begin{equation}\label{Ahanders}
A_h(x,v) = {\mathcal D}_F L_h(x,v) \cdot v\; .
\end{equation}

\begin{Def}\label{singularpunkt}
For a smooth manifold $M$, a $\Ce^\infty$-function $h: T^*M\ra \R$ and $E\in \R$, we define the singular set
$S_h(E)$ by
\[ S_h(E) := \{x\in M\;|\;  h(x,0) \geq E\, \}\, . \]
Since
$h(.,0)$ is continuous, $S_h(E)$ is closed. Thus
$\widetilde{M}:=M\setminus S_h(E)$ is again a smooth manifold.\\
\end{Def}

For $h$ hyperregular, we shall now introduce an associated
Finsler function on $\widetilde{M}$.

\begin{theo}\label{listfinsler}
Let $M$ and $h\in{\mathscr C}^\infty \left(T^*M\right)$
satisfy Hypothesis \ref{hypfins} and let $E$, $S_h(E)$ and $\widetilde{M}$ be as described in Definition
\ref{singularpunkt}.
\ben
\item[i)] Then there exists a strictly fibre preserving ${\mathscr C}^\infty$-function
$ \tau_{E}  : S\widetilde{M} \rightarrow T{\widetilde M}$, which is uniquely determined
by the condition
\begin{equation}\label{taucond}
 h\circ\left({\mathcal D}_Fh\right)^{-1}\circ \tau_{E} = E \; .
\end{equation}
\item[ii)]
Let $\tilde{\tau}_{E}:=\tau_{E}\circ\pi_S:T{\widetilde M}\ra T\widetilde{M}$ and
let $\ell_{h,E}:T\widetilde{M} \rightarrow \R$ be defined by
\[ \ell_{h,E}(x,v):= \left({\mathcal D}_Fh\right)^{-1} \circ\tilde{\tau}_{E}(x,v)\,\cdot\, v \, . \]
Then $\ell_{h,E}$ is an absolute homogeneous Finsler function on $\widetilde{M}$.
\item[iii)]
For any regular curve $\gamma:[a,b]\ra \widetilde{M}$, there exists a unique ${\mathscr C}^1$-function
$\lambda:[a,b]\ra \R_+$ such that
\begin{equation}\label{tauundlambda}
\tilde{\tau}_{E}(\gamma(t),\dot{\gamma}(t)) = (\gamma(t), \lambda(t)\dot{\gamma}(t))\; .
\end{equation}
\een
\end{theo}

\begin{rem}\label{energieschale}
\ben
\item Since $\ell_{h,E}$ is defined on
$\widetilde{M} = M\setminus S(E)$, we call $(M,\ell_{h,E})$
a Finsler manifold with singularities.
\item If we continuously extend $\ell_{h,E}$ from $\widetilde{M}$ to $M$ by setting $\ell_{h,E}(x,v)=0$ for $x\in S(E)$,
the associated distance $d_\ell$ is well defined on all of $M$. Nevertheless
contrary to the case of a Finsler manifold without singularities (as described for example in
Bao-Chern-Shen \cite{bao}), the geodesic curves with respect to $\ell_{h,E}$ may have kinks at the
``singular points'', which are the connected components of $S_h(E)$.
\item Geometrically, the function $\tilde{\tau}_{E}$ projects an
element $(x,v)$ of the tangent bundle $T\widetilde{M}$ to an element $(x,\lambda v)$
in the $(2d-1)$-dimensional
submanifold $\mathcal{E}= E_h^{-1}(E)$.
\item Schematically the functions occurring in Theorem \ref{listfinsler} are illustrated in the following diagram.
\[
\begin{diagram}
\node[3]{\R}\\
\node{T_x\widetilde M}\arrow{nee,t}{\ell_{h,E}}\arrow{s,l}{(\pi_S, \id)}
\arrow{e} \node{{\mathcal E}\times T_x\widetilde M}
\arrow{e}\node{h^{-1}(E)\times T_x\widetilde M}
\arrow{n}\\
\node{S_x\widetilde M\times T_x\widetilde M}\arrow{ne,b}{\tau_{E}\times\id}\\
\node{v}\arrow{e,t}{(\tilde{\tau}_{E}, \id)}
\node{(\tilde{v},v)}
\arrow{e,t}{(\mathcal{D}_Fh)^{-1}\times \id}
\node{(\xi_h(\tilde{v}),v)}
\end{diagram}
\]
\item With the notation \eqref{xivonv}, $\ell_{h,E}(x,v)$ can be written as
\begin{equation}\label{lalspaar}
\ell_{h,E}(x,v) = \xi_h(x,\tilde{v})\cdot v\quad\text{where}\quad (x,\tilde{v}) = \tilde{\tau}_{E}(x,v)\in
{\mathcal
E} \; .
\end{equation}
\een
\end{rem}

To prove Theorem \ref{listfinsler}, we need the following lemma.

\begin{Lem}\label{energiemonoton}
In the setting of Theorem \ref{listfinsler} fix $x\in\widetilde{M}$ and $u\in T_x\widetilde{M}\setminus \{0\}$.
Then
for $E_h:TM\ra \R$ defined by (\ref{energiefunktion}), the function
\[ E_u:[0,\infty)\rightarrow \R\; , \qquad E_u(\lambda) := E_h (x,\lambda u)\]
is strictly increasing with $\frac{d}{d\lambda}E_u>0$ for $\lambda>0$.
Furthermore $E_u(0)\leq E$ and $\lim_{\lambda\to\infty} E_u(\lambda) = \infty$.
\end{Lem}

\begin{proof}[Proof of Lemma \ref{energiemonoton}]

Since $h_x$ is even, $D h_x(0) = 0$, thus $v_h(x,0) = 0$, $\xi_h(x,0)=0$ and
\begin{equation}\label{Eunull}
E_u(0) = E_h(x,0) = h(x,0) < E\;  .
\end{equation}
To show that $E_u$ is strictly increasing, we will analyze the derivative of
$E_u$ for $\lambda >0$.\\
By definition ${\mathcal D}_F h (x,\xi) = Dh_x(\xi)$, thus
\begin{equation}\label{energie1}
\frac{dE_u}{d\lambda}|_{\lambda} = D h_x|_{(Dh_x)^{-1}(\lambda u)}
\cdot D \left(Dh_x\right)^{-1}|_{\lambda u} (u) \; .
\end{equation}
We notice that
\begin{equation}\label{energie2}
D h_x |_{(Dh_x)^{-1}(\lambda u)} = \lambda u
\end{equation}
and for $L_h$ defined in \eqref{lagrangefunktion} it follows from \eqref{lagrangehamilton} that
\begin{equation}\label{energie2a}
D \left(Dh_x\right)^{-1}|_{\lambda u} (u) = D \left(DL_{h,x}\right)|_{\lambda u} (u)\; .
\end{equation}
Inserting \eqref{energie2a} and \eqref{energie2} in \eqref{energie1}
yields
\begin{equation}\label{energie4}
\frac{dE_u}{d\lambda}|_{\lambda} = \lambda u \cdot D \left(DL_{h,x}\right)|_{\lambda u} (u) =
\lambda\, D^2 L_{h,x}|_{\lambda u} (u,u) \; ,
\end{equation}
where we identify linear maps from
$T_xM$ to $T_x^*M$ with bilinear forms on $T_xM$.
If $h$ is strictly convex in each fibre, the same is true for $L_h$. Therefore by \eqref{energie4} the
first derivative of $E_u$ is strictly positive for $\lambda \in (0,\infty)$ and thus $E_u$ is
strictly increasing.

The fact that $\lim_{\lambda\to\infty}E_u(\lambda) = \infty$ can be seen
as follows.
From the strict convexity of $h$ and since $h(x,\xi) \geq h(x,0)$ for all $\xi\in T^*M$,
it follows that $\lim_{|\xi|\to\infty}h(x,\xi) = \infty$. Since
$h$ is hyperregular, ${\mathcal D}_Fh(x,.) = Dh_x : T_xM  \ra T_x^*M$ is a global diffeomorphism.
Thus for any norm $\|.\|_{T_xM}$ on $T_xM$ and the induced norm $\|.\|_{T_x^*M}$ on $T_x^*M$, we have
$\|Dh_x(v_n)\|_{T_x^*M} \to \infty$ for any sequence $(v_n)$ in $T_xM$ satisfying $\|v_n\|_{T_xM}\to \infty$
(any global diffeomorphism is proper).
Thus
\[ \lim_{\lambda\to\infty}E_u(\lambda) = \lim_{\lambda\to\infty}h(x,\xi_h(x,\lambda u)) =
\lim_{\|\xi\|\to\infty}h(x,\xi) = \infty\; . \]
\end{proof}

\begin{proof}[Proof of Theorem \ref{listfinsler}]

i)
From Lemma \ref{energiemonoton} it follows that for fixed $x\in \widetilde{M}$, each ray
$[x,u]\in S\widetilde{M}$
intersects the hypersurface
${\mathcal E}_x := E_h^{-1}(E)\cap T_x\widetilde{M}$
in exactly one point $v$, i.e.
for each ray $[x,u]$ there is exactly one point $(x,v)\in T_x\widetilde{M}$ such
that $E_h(x,v) = E$. Thus $[x,u] \mapsto (x, v)$ defines a map $\tau_{E}$, which is uniquely determined
by \eqref{taucond}.
Clearly
$\tau_{E}$ is strictly fibre preserving.

To analyze the regularity of $\tau_{E}$, we will use the Implicit Function Theorem.
We remark that there exists an isomorphism $\psi: T\widetilde{M} \setminus\{0\} \ra S\widetilde{M} \times \R_+$, such
that $\psi^{-1} (\{[x,u]\} \times \R_+) = [x,u]\subset T\widetilde{M}\setminus\{0\}$. Setting
$\hat{E}_h:= E_h \circ \psi^{-1}$ and $\hat{\tau}_{E} := \psi \circ \tau_{E}$, we get
$\hat{\tau}_{E} (s) = (s,\lambda (s)), s\in S\widetilde{M},$ for some $\lambda: S\widetilde{M}\ra \R_+$ and
therefore $\hat{E}_h(s,\lambda (s)) = E$.
By Lemma \ref{energiemonoton}, $\frac{d \hat{E}_h}{d\lambda}(s_0,\lambda_0)>0$ for all
$(s_0,\lambda_0)\in S\widetilde{M}\times \R_+$.
Smoothness of $\lambda$ and thus of $\tau_{E}$ now follows
from the Implicit Function Theorem.\\
iii) By i), for $t\in[a,b]$ fixed, there is a unique $\lambda(t)\in\R_+$ with \eqref{tauundlambda}.
Since $\tau_{E}$ is smooth (as composition of smooth maps), $t\mapsto (\gamma(t),\lambda(t)\dot{\gamma}(t))$
is a $\Ce^1$-curve. Thus $\lambda\in\Ce^1([a,b],\R)$.\\
ii) To show that $\ell_{h,E}: T\widetilde{M} \ra \R$ is a Finsler function on $\widetilde{M}$, we check the
defining properties.
\ben
\item[1)] The regularity $\ell_{h,E}\in \Ce^\infty (T\widetilde{M}\setminus \{0\})$ follows from the fact that
$h$ is hyperregular and the function $\tilde{\tau}_{E}$ is
$\Ce^\infty$.
\item[2)] To show $\ell_{h,E}(x,\lambda v) = \lambda \ell_{h,E}(x,v)$, ($\lambda>0$), we notice that
by construction $\tilde{\tau}_{E} (x,\lambda v) = \tilde{\tau}_{E} (x,v)$ for any $\lambda>0$. Thus
$\left({\mathcal D}_Fh\right)^{-1}\circ \tilde{\tau}_{E}$ is homogeneous of
order zero in each fibre. Since $\xi\cdot v$ is bilinear, it follows that
\[\ell_{h,E}(x,\lambda v) = \left({\mathcal D}_Fh\right)^{-1}\circ \tilde{\tau}_{E}(x,\lambda v)\,\cdot\,
 \lambda v
 = \lambda \ell_{h,E}(x,v) \, . \]
\item[3)] We start showing that $\ell_{h,E}(x,v)>0$, ($v\neq 0$). To this end we
define
\[ a_h=A_h\circ {\mathcal D}_Fh \,:\, T^*M \ra \R\; , \quad \quad
a_h(x,\xi) =  \xi\cdot {\mathcal D}_Fh
(x,\xi)\, .\]
Since $h$ was assumed to be strictly convex in each fibre, one obtains from \eqref{energie8}
\[ (\xi-\eta) \cdot \left( {\mathcal D}_Fh (x,\xi) - {\mathcal D}_Fh (x,\eta)\right) > 0\: ,\quad
\xi,\eta\in T^*_xM, \;\eta\neq \xi \; . \]
Therefore choosing $\xi = - \eta$ and using that $h$ is even in each fibre (thus ${\mathcal D}_Fh $ is odd) yields
\begin{equation}\label{vierah}
2\xi \cdot \left( {\mathcal D}_Fh(x,\xi)- {\mathcal D}_Fh(x,-\xi)\right) =
4\xi\cdot  {\mathcal D}_Fh (x,\xi) = 4 a_h(x,\xi) > 0 \,,\quad\text{for}\quad \xi\neq 0\; .
\end{equation}
Since $h$ is even and strictly convex, it takes its absolute
minimum at $\xi=0$ and thus ${\mathcal D}_Fh(x,0) = 0$. Since furthermore
${\mathcal D}_Fh$ is a global diffeomorphism, we get $\left({\mathcal D}_Fh\right)^{-1} (x,v) \neq 0$
for $v\neq 0$. By \eqref{vierah}
\begin{equation}\label{Gxpos}
A_h(x,v) = a_h(x,\left({\mathcal D}_Fh\right)^{-1}(x,v)) > 0\, , \qquad v\neq 0 \; .
\end{equation}
Setting $\tilde{\tau}_{E}(x,v) =: (x,\tilde{v})$, it follows from the fact that $\tau_{E}$ is
strictly fibre preserving that there exists a $\lambda >0$ such that $v = \lambda \tilde{v}$. Thus
by \eqref{Gxpos} for $v\neq 0$
\begin{equation}\label{g14}
\ell_{h,E}(x,v) =
 \left({\mathcal D}_Fh\right)^{-1}(x,\tilde{v})\,\cdot\,\lambda \tilde{v}
= \lambda A_h(x,\tilde{v}) > 0\; .
\end{equation}

To show that the matrix $g$ is positive definite, we set $\ell_x(v):= \ell_{h,E}(x,v): T_x\widetilde{M} \ra \R$.
Then we have to show that $g_{(x,v)}(w,w) = D^2(\frac{1}{2} \ell_x^2)|_{v}(w,w)>0$ for all $v,w\in T_x\widetilde{M}\setminus \{0\}$.\\
We first remark that 
\begin{align}\label{g1}
 g_{(x,v)}(w,w) = T_1(w) + T_2(w)\qquad\text{where}\\
 T_1(w) = \ell_x(v) D^2\ell_x|_v (w,w) \quad\text{and}\quad T_2(w) = \left(D\ell_x|_v (w)\right)^2\; . \nonumber
 \end{align}
By the definition of $\ell_x$ and \eqref{lagrangehamilton} it follows that
\begin{equation}\label{g10} 
D\ell_x|_v(w) = D^2L_x|_{\tilde{\tau}_{E}(v)} \left(D\tilde{\tau}_{E}|_v (w) , v\right) +
\left(Dh_x\right)^{-1} \circ \tilde{\tau}_{E} (v) \cdot w \; . 
\end{equation}
To analyze $D\tilde{\tau}_{E}$, we use that by Lemma \ref{energiemonoton} the function $E_{h,x}$ is strictly increasing
in each fibre. Therefore, analogue to the proof of i), there exists a smooth function $\mu_x: T_xM \ra (0,\infty)$ such that 
$\tilde{\tau}_{E}(v) = \mu_x(v) v$ and thus 
\begin{equation}\label{g8} 
D\tilde{\tau}_{E}|_v (w)= v\cdot D\mu|_v (w) + \mu(v) w \; .
\end{equation}
Since 
\[ E = E_{h,x}(\tilde{\tau}_{E})(v) = E_{h,x}\circ G_x(v,\mu(v)) \quad\text{where}\quad G_x(v,\mu(v)) := \mu(v) v \]
we have $D_v(E_{h,x}\circ G_x)(v,\mu(v)) = 0$, leading to 
\begin{equation}\label{g4}
D\mu|_v (w) = -\left(D_2(E_{h,x}\circ G_x)|_{(v,\mu(v))}\right)^{-1} D_1( E_{h,x}\circ G_x)|_{(v,\mu(v))} (w)
\end{equation}
Since by the definition of $E_h$ in \eqref{energiefunktion} and again by \eqref{lagrangehamilton} we have 
\begin{equation}\label{g5}
DE_{h,x}|_v(w) = Dh_x|_{(Dh_x)^{-1}(v)} \cdot D(Dh_x)^{-1}|_v (w)= D^2L|_v(v,w)\, ,
\end{equation}
it follows that
\begin{equation}\label{g3}
D_2(E_{h,x}\circ G_x)(v,\mu) = DE_{h,x}|_{\mu v} (v) = \mu D^2L|_v (v,v) \, .
\end{equation}

By \eqref{g5} it follows at once that
\begin{equation}\label{g6}
D_1( E_{h,x}\circ G_x)|_{(v,\mu(v))} (w) = DE_{h,x}|_{\mu(v)v} D_1G_x|_{(\mu(v),v)} (w) = \mu^2(v) D^2L_x|_{\mu(v)v} (v,w)\; .
\end{equation}
Since $L$ is strictly convex, the symmetric bilinear form $D^2L_x|_{\mu(v)v} (.,.)$ is positive definite and
therefore defines a scalar product 
\begin{equation}\label{g13}
\langle v, w\rangle_L :=  D^2L_x|_{\mu(v)v} (v,w) \quad\text{with associated norm} \quad \|v\|_L 
\end{equation}
Inserting \eqref{g6} and \eqref{g3} in \eqref{g4} and using \eqref{g13} gives
\begin{equation}\label{g7}
D\mu|_v (w)= -\frac{\mu(v)}{\|v\|^2_L} \langle v,w\rangle_L\; .  
\end{equation}
Thus inserting \eqref{g7} in \eqref{g8} and the resulting term in \eqref{g10} yields
\begin{align}
D\ell_x|_v (w) &= \left\langle \left[-v \frac{\mu(v)}{\|v\|_L^2} \langle v,w\rangle_L + \mu(v) w\right], 
v \right\rangle_L + \left(Dh_x\right)^{-1}|_{\mu(v)v} \cdot w \nonumber\\
&=  \left(Dh_x\right)^{-1}|_{\mu(v)v} \cdot w   \label{g9}
\end{align}
Using \eqref{g9} and \eqref{lagrangehamilton} we get
\[
D^2\ell_x|_v (w,w) = \left\langle \left(v D\mu|_v (w) + \mu(v) w\right) , w\right\rangle_L
\]
and inserting \eqref{g4} gives
\begin{equation}\label{g12}
D^2\ell_x|_v (w,w) = \mu(v) \left\{ \langle w,w\rangle_L - \frac{(\langle v,w\rangle_L)^2}{\|v\|^2_L} \right\} \geq 0\; , 
\end{equation}
where the last estimate follows from the Cauchy-Schwarz Inequality.
Since $T_2$ is quadratic, \eqref{g12} together with \eqref{g14} gives 
\[ D^2\left(\frac{1}{2} \ell_x^2\right)|_{v} (w,w) \geq 0 \, , \qquad v,w\in T_x\widetilde{M} \; . \]
To prove the strict positivity, we now fix
$x\in\widetilde{M}$ and $v\in T_x\widetilde{M}\setminus\{0\}$. Assuming 
\begin{equation}\label{g15}
g_{(x,v)} (w,w) = 0\quad\text{for}\quad w\in T_x\widetilde{M}\; ,
\end{equation}
we have to show that $w=0$.
By \eqref{g1} it follows from \eqref{g15} that $T_1(w)=0$ and $T_2(w) =0$. We have already seen in \eqref{g14} that 
$\ell_{h,E}(x,v)>0$ for $v\neq 0$, thus $T_1(w)=0$ implies $g_{(x,v)} (w,w) = 0$, leading by \eqref{g12}
and the Cauchy-Schwarz-inequality to
\[  w=\eta v \quad \text{for some}\quad \eta\in\R \; .\]
Inserting this in $T_2$, the homogeneity of $\ell_x$ shows
\[ 0 = D\ell_x|_v (\eta v) = \eta D\ell_x|_v (v) = \eta \ell_x(v) \]
and thus by the positivity of $\ell_x$ we get $\eta=0$ and thus $w=0$.
\item[4)]
It remains to show that $\ell_{h,E}$ is absolute homogeneous of order one.
Since $h$ is even in each fibre, $\mathcal{D}_F h$ and $\left(\mathcal{D}_Fh\right)^{-1}$ are odd.
Thus for $(x,v)\in{\mathcal E}_x$
\begin{equation}\label{Exsymm}
h\circ \left(\mathcal{D}_Fh\right)^{-1}(x,-v) =
h\circ \left(\mathcal{D}_Fh\right)^{-1}(x,v) = E
\end{equation}
and \eqref{Exsymm} yields
\begin{equation}\label{plusimpmin}
(x,v)\in\mathcal{E} \qquad \Longrightarrow \qquad (x,-v)\in{\mathcal E}\, .
\end{equation}
Since $\tau_{E}$ is strictly fibre preserving, we have for $(x,v)\in T_x\widetilde{M}$ and some $\lambda>0$, using
\eqref{plusimpmin}
\begin{equation}\label{tautildeodd}
\tilde{\tau}_{E} (-v) = \lambda(-v) = - \tilde{\tau}_{E} (v)\; .
\end{equation}
By the fact that $\left(\mathcal{D}_Fh\right)^{-1}$ is odd and \eqref{tautildeodd}
we can conclude that
\begin{equation}\label{lgerade}
\ell_{h,E}(-v) =  \left(\mathcal{D}_Fh\right)^{-1}\circ\tilde{\tau}_{E}(-v)\,\cdot\,(-v) =
    \left(\mathcal{D}_Fh\right)^{-1}\circ\tilde{\tau}_{E}(v)\,\cdot\, v = \ell_{h,E}(v)\;   .
\end{equation}
By \eqref{lgerade}, $\ell_{h,E}$ is even in each fibre and thus for any $\lambda\in \R$
\[
\ell_{h,E}(x, \lambda v) = \ell_{h,E}(x, |\lambda| v) = |\lambda| \ell_{h,E}(x,v) \; . \]
\een
\end{proof}

\subsection{Proof of Theorem 1.4}\label{geodesic}

 Step 4 of our proof is
adapted from Abraham-Marsden \cite{abma} and uses the Maupertuis principle (at least implicitly).\\

{\sl Step 1:}\\
We will show that
\begin{multline}\label{eumau20}
\Gamma (x_1,x_2,[a,b],E):= \{ (\gamma,\alpha)\; | \; \alpha: [a,b]\ra \R\quad\text{is}\quad\Ce^2,\quad
\dnt\alpha >0\, ,\quad \alpha(a) = 0\, ,\\
\gamma\in \Gamma_{0,\alpha(b)}(x_1,x_2)
\quad\text{such that}\quad E_h(\gamma(\alpha(t)),\dot{\gamma}(\alpha(t)) = E\quad\text{for all}\quad t\in[a,b]\, \}
\, ,
\end{multline}
is a Banach manifold, where $\Gamma_{a,b}(x_1,x_2)$ was introduced
in \eqref{gammaab}.\\
$\Gamma (x_1,x_2,[a,b],E)$ is the set of all pairs $(\gamma,\alpha)$, where $\gamma$ is a
regular curve on $\widetilde{M}$ joining
the points $x_1$ and $x_2$ and $\alpha$ is a
change of parameter, ensuring that the curve
$(\gamma\circ\alpha, \dot{\gamma}\circ\alpha)\in TM$ (which is not equal to the lifted curve
$(\gamma\circ \alpha, \frac{d}{dt} (\gamma\circ \alpha))$) lies on the energy shell
${\mathcal E} = E_h^{-1}(E)$.

Set
$A:= \{ \alpha:[a,b]\ra \R\,|\, \dnt\alpha >0\quad\text{and}\quad \alpha(a) = 0 \}$
and denote by $\Gamma_{0,\infty}$ the space of all regular curves $\gamma:[0,\infty) \ra \widetilde{M}$. Then
$\Gamma_{0,\infty}\times A$ is a Banach manifold. We consider the $\Ce^1$-mapping
\[ g: \Gamma_{0,\infty}\times A \ra \widetilde{M}\times \widetilde{M} \, , \quad (\gamma, \alpha) \mapsto
(\gamma\circ\alpha (a), \gamma\circ\alpha (b))\; . \]
Then $(x_1,x_2)\in \widetilde{M}\times \widetilde{M}$ is a regular value of $g$ and
\[ \Gamma ([a,b],x_1,x_2):= g^{-1}(x_1,x_2) = \{ (\gamma,\alpha)\in\Gamma_{0,\infty}\times A\, |\,
\gamma(\alpha(a))=x_1, \gamma(\alpha(b))=x_2\} \]
is a submanifold of $\Gamma_{0,\infty}\times A$.
This follows from the fact that the Inverse
Function Theorem holds in Banach manifolds (see Hamilton \cite{hamilton}).
We introduce
\[ \tilde{E}_h: \Ce^1([a,b], T\widetilde{M})\ra \Ce^1([a,b],\R)\,, \quad \tilde{E}_h(\eta)(t):= E_h(\eta(t))\]
and for
\[
\Gamma^{T\widetilde{M}}_{a,b}(x_1,x_2):= \{(\gamma, k\,\dot{\gamma})\in \Ce^1([a,b], T\widetilde{M})\,|\, \gamma\in\Gamma_{a,b}(x_1,x_2),\;
k\in\Ce^1([a,b],\R_+)\} \]
we set
\[ \Phi: \Gamma ([a,b],x_1,x_2) \ra \Gamma^{T\widetilde{M}}_{a,b}(x_1,x_2)\, , \quad (\gamma,\alpha)\mapsto
(\gamma\circ\alpha, \dot{\gamma}\circ \alpha)\; . \]
Then $\Phi$ is a diffeomorphism. In fact, by a straightforward calculation it is bijective, with inverse
$\Phi^{-1}(\eta, k\dot{\eta}) = (\eta \circ \alpha^{-1}, \alpha)$, where $\alpha(t) = \int_a^t \left(k(s)\right)^{-1}\, ds$.

Identifying $E$ with the constant function $E(t)=E$, we obtain
$\Gamma(x_1,x_2,[a,b],E) = f^{-1}(E)$ for $f:= \tilde{E}_h\circ \Phi$.
To show that $E$ is a regular value of $f$ it is sufficient to show that it is a regular value of
$\tilde{E}_h$, i.e.,
that for each $v\in \Ce^1([a,b],\R)$ (considered as a vector field along
$E \in \Ce^1([a,b],\R)$), there is a vector field $X$ along $\eta\in \tilde{E}_h^{-1}(E)$
with $d\tilde{E}_h|_\eta X = v$. Note that $\eta(t) = (x(t), v(t))$ with $v(t)\neq 0$.

Since $DE_h(x,v)\neq 0$ for $v\neq 0$, there is a covering of $[a,b]$ by open intervals $I_j, j\in J$ and vector fields $X_j$
along $\eta|_{I_j}$ with $DE_h|_{\eta(t)} X_j(t) = v(t)$ for all $t\in I_j$.
Choosing a partition of unity $(\chi_j)$ subordinate to $(I_j)$, we set $X(t) = \sum_{j\in J} \chi_j(t)X_j(t)$.
Then $d\tilde{E}_h|_\eta X = v$. Thus $E$ is a regular value of $f$ and
$\Gamma ([a,b],x_1,x_2,E)$ is a Banach manifold.\\

{\sl Step 2:}\\
We construct a diffeomorphism
$b_{E}: \Gamma_{a,b}(x_1,x_2) \ra \Gamma (x_1,x_2,[a,b],E)$.

By Theorem \ref{listfinsler}, there exists for any $\eta\in \Gamma_{a,b}(x_1,x_2)$
a unique $\Ce^1$-function
$\lambda: [a,b] \ra \R_+$ such that
\[  E_h(\eta(t), \lambda (t)\dot{\eta}(t)) = E\; .\]
Set
\begin{equation}\label{pairdef}
\alpha (t):=\int_a^t \frac{1}{\lambda(s)}\,ds\quad\text{and}\quad
\gamma = \eta\circ\alpha^{-1}:[0,\alpha(b)]\ra \widetilde{M}\; ,
\end{equation}
then $\alpha:[a,b]\ra \R$ with $\dot{\alpha}>0$. From
$\dot{\eta}(t) = \dot{\gamma}(\alpha (t)) \cdot \dot{\alpha}(t)$
it follows that
\begin{equation}\label{eumau1}
E_h(\gamma(\alpha(t)), \dot{\gamma}(\alpha(t)))  =
E_h(\eta(t), \lambda(t)\dot{\eta}(t)) = E\; ,
\end{equation}
i.e. $(\gamma, \alpha)\in \Gamma(x_1,x_2,[a,b],E)$.
We can conclude that there is a
bijection between Banach manifolds given by
\begin{align}
b_{E}: \Gamma_{a,b}(x_1, x_2)
&\ra \Gamma (x_1,x_2,[a,b],E) \nonumber\\
\label{eumau23}
b_{E}(\eta) &= (\eta\circ\alpha^{-1}, \alpha)\quad
\text{with}\quad \alpha (t):=\int_a^t \left(\lambda(s)\right)^{-1}\,ds\quad\text{for}\quad
\tilde{\tau}_{E}(\eta,\dot{\eta})=(\eta,\lambda \dot{\eta}) \; .
\end{align}
On the other hand, if we start with $(\gamma,\alpha)\in\Gamma(x_1,x_2,[a,b],E)$, then
$E_h(\gamma(s),\dot{\gamma}(s))=E$ with $s=\alpha(t)$. Setting
$\eta:= \gamma\circ\alpha : [a,b]\ra \widetilde{M}$ it follows from $\dot{\gamma}(s) =
\dot{\eta}(t)\,(\dot{\alpha}(t))^{-1}$ that
\begin{equation}\label{eumau16}
E_h(\eta(t), (\dot{\alpha}(t))^{-1} \,\dot{\eta}(t)) = E
\qquad\text{and thus}\qquad
\lambda(t)=(\dot{\alpha}(t))^{-1}\; .
\end{equation}
Thus the inverse function $b_{E}^{-1}$ is given by $b_{E}^{-1}(\gamma,\alpha) = \gamma\circ \alpha$.\\

{\sl Step 3:}\\
We show that the critical points of the length functional $s_{\ell_h}$ defined in Definition \ref{Finslerman2}
(i.e., the geodesics of $\ell_h$)
are in bijection with the critical points of the action integral
\begin{equation}\label{eumau3}
I:\Gamma(x_1,x_2,[a,b],E)\ra \R\, , \quad
I(\gamma,\alpha):= \int_{\alpha(a)}^{\alpha(b)} A_h(\gamma(s),\dot{\gamma}(s))\, ds \; ,
\end{equation}
where $A_h$ denotes the action with respect to $h$ defined in \eqref{kanpaarung}.
Setting $s=\alpha(t)$ and using \eqref{xivonv} gives
\begin{equation}\label{eumau2}
\int_{\alpha(a)}^{\alpha(b)} A_h(\gamma(s),\dot{\gamma}(s))\, ds
%\int_{\alpha(a)}^{\alpha(b)} \xi_h(\gamma(s),\dot{\gamma}(s))\cdot \dot{\gamma}(s)\, ds \\
=\int_{a}^{b} \xi_h(\gamma(\alpha(t)),\dot{\gamma}(\alpha(t)))\cdot \dot{\gamma}(\alpha(t))\dot{\alpha}(t)\,
dt\; .
\end{equation}
Setting $\eta(t)=\gamma(\alpha(t))$ and using \eqref{eumau16} and the definition of $\ell_{h}$ and $s_{\ell_h}$,
we obtain from \eqref{eumau3} and \eqref{eumau2}
\begin{align*}
I(\gamma,\alpha) &= \int_{a}^{b} \xi_h(\eta(t),\dot{\eta}(t)(\dot{\alpha}(t))^{-1})\cdot \dot{\eta}(t)\, dt =
\int_{a}^{b} \xi_h\circ \tilde{\tau}_{E}(\eta(t),\dot{\eta}(t))\cdot \dot{\eta}(t)\, dt \\
&= \int_{a}^{b} \ell_{h,E}(\eta(t),\dot{\eta}(t))\, dt = s_{\ell_{h}}(\eta)\; .
\end{align*}
Since $b_{E}(\gamma\circ \alpha) = (\gamma,\alpha)$, it follows that
\begin{equation}\label{eumau19}
s_{\ell_h} = I \circ b_{E}\quad\text{and thus}\quad ds_{\ell_{h}}|_\eta = dI|_{b_{E}(\eta)} \circ
db_{E}|_\eta\; .
\end{equation}
Since $b_{E}$ is a diffeomorphism, we get
\begin{equation}\label{eumau40}
ds_{\ell_h}|_\eta = 0 \qquad \Longleftrightarrow \qquad dI|_{b_{E}(\eta)} = 0 \qquad \eta\in\Gamma_{a,b}(x_1,x_2)\; .
\end{equation}

{\sl Step 4:}\\
We show (a).

We set $\gamma_0(a)=x_1$, $\gamma_0(b)=x_2$. If $\gamma_0$ is a base integral curve of the Hamiltonian vector field $X_h$ with
$E_h(\gamma_0(t),\dot{\gamma}_0(t))=E$ for all $t\in[a,b]$, then
$b_{E}(\gamma_0) = (\gamma_0,\id)$, where $\id:[a,b]\ra[a,b]$ is defined by
$\id(t) = t$.

Thus by \eqref{eumau40} it remains to show that $dI|_{(\gamma_0, \id)} = 0$
for any base integral curve $\gamma_0\in \Gamma_{a,b}(x_1,x_2)$ of the
Hamiltonian vector field $X_h$ with energy $E$.
The tangent space of $\Gamma (x_1,x_2,[a,b],E)$ at a point
$(\gamma,\alpha)$ can be described by use of variations as
\begin{eqnarray}\label{tangentgamma}
T_{(\gamma,\alpha)}\Gamma (x_1,x_2,[a,b],E) = \{
\partial_u (\gamma,\alpha)_\delta|_{u=0}\,|\,
(\gamma,\alpha)_\delta :(-\delta,\delta)\ra \Gamma(x_1,x_2,[a,b],E)\\
\text{is}\quad\Ce^2\quad\text{with}
\quad (\gamma,\alpha)_\delta(0)=(\gamma,\alpha)\}\; .\nonumber
\end{eqnarray}

We start analyzing $dI|_{(\gamma,\alpha)}$. We use the
notation $(\gamma(\,.\,),\alpha(\,.\,))_\delta (u) =: (\gamma_\delta(\,.\,,u), \alpha_\delta (\,.\,,u))$.
From \eqref{tangentgamma} and \eqref{eumau20}
it follows that
$\alpha_\delta (a,u) = 0$. Furthermore
$\gamma_\delta(0,u) = x_1$ and
$\gamma_\delta(\alpha_\delta(b,u),u) = x_2$ for all $u\in(-\delta,\delta)$.
This leads to
\begin{equation}\label{eumau9}
\frac{d}{du} \gamma_\delta (0,u) = 0 = \frac{d}{du} \gamma_\delta (\alpha_\delta (b,u),u) \; .
\end{equation}
Since $A_h = L_h + E_h$ by the definition \eqref{lagrangefunktion} of the Lagrange function $L_h$,
it follows from the definition
\eqref{eumau20} of $\Gamma(x_1,x_2,[a,b],E)$ that
\begin{equation}\label{eumau13}
A_h(\gamma_\delta(\alpha_\delta(t,u),u),\dot{\gamma}_\delta(\alpha_\delta(t,u),u))=
L_h(\gamma_\delta(\alpha_\delta(t,u),u),\dot{\gamma}_\delta(\alpha_\delta(t,u),u)) + E\; ,
\end{equation}
thus the definition \eqref{eumau3} of $I$ and \eqref{eumau13} yield
\begin{multline}\label{eumau10}
dI|_{(\gamma,\alpha)}\left(\partial_u(\gamma,\alpha)_\delta|_{u=0}\right) =
\partial_u I\left((\gamma_\delta ,\alpha_\delta)\right)|_{u=0} \\
\left. = \frac{d}{du}  \int_{\alpha_\delta(a,u)}^{\alpha_\delta(b,u)}
\left(L_h(\gamma_\delta(s,u),\dot{\gamma}_\delta(s,u)) + E\right)\, ds \right|_{u=0}.
\end{multline}
We get using $\gamma_\delta(t,0) = \gamma(t)$ and $\alpha_\delta(t,0) = \alpha(t)$
\begin{multline}\label{eumau10a}
\left. \frac{d}{du}  \int_{\alpha_\delta(a,u)}^{\alpha_\delta(b,u)}
\left(L_h(\gamma_\delta(s,u),\dot{\gamma}_\delta(s,u)) + E\right)\, ds \right|_{u=0} \\
=
\left[\left(L_h(\gamma(\alpha(t)),\dot{\gamma}(\alpha(t))) +
E\right)\cdot
\partial_u\alpha_\delta|_{u=0}(t)\right]_a^b  \\
+ \int_{0}^{\alpha(b)}
\left. \frac{d}{du} L_h(\gamma_\delta(s,u),\dot{\gamma}_\delta(s,u))\right|_{u=0}\; ds\; .
\end{multline}
For the integrand on the right hand side of \eqref{eumau10a} we get
\begin{multline}\label{eumau41}
\left. \frac{d}{du} L_h(\gamma_\delta(s,u),\dot{\gamma}_\delta(s,u))\right|_{u=0} =
D_{\gamma}L_h(\gamma(s),\dot{\gamma}(s))\cdot
\partial_u\gamma_\delta|_{u=0}(s)  \\
+ D_{\dot{\gamma}}L_h(\gamma(s),\dot{\gamma}(s))
\partial_u\dot{\gamma}_\delta|_{u=0}(s)\;  .
\end{multline}
Since $\partial_u\dot{\gamma}_\delta|_{u=0} (s) =\partial_s \partial_u \gamma_\delta|_{u=0} (s)$,
integration by parts for the second summand on the right hand side of
\eqref{eumau41} gives
\begin{multline}\label{eumau11}
\int_{\alpha(a)}^{\alpha(b)}
\left.\frac{d}{du}L_h(\gamma_\delta(s,u),\dot{\gamma}_\delta(s,u))\right|_{u=0}\; ds =
\left[D_{\dot{\gamma}}L_h(\gamma(s),\dot{\gamma}(s))\cdot
\partial_u \gamma_\delta(s,u)|_{u=0}\right]_{0}^{\alpha(b)}  \\
- \int_{0}^{\alpha(b)}\left(
 D_{\gamma}L_h(\gamma(s),\dot{\gamma}(s))
 + \frac{d}{ds}
 D_{\dot{\gamma}}L_h(\gamma(s),\dot{\gamma}(s))\right)\cdot
\partial_u \gamma_\delta|_{u=0}(s) \, ds \; .
\end{multline}
It follows from \eqref{eumau9} that
\begin{equation}\label{eumau21}
\partial_u \gamma_\delta(\alpha(t),u)|_{u=0}
= - \dot{\gamma}(\alpha(t))  \partial_u\alpha_\delta |_{u=0}(t)\, , \quad t=a,b \; .
\end{equation}
Since by \eqref{Ahanders} we have
$A_h(\gamma,\dot{\gamma}) =  D_{\dot{\gamma}}L_h(\gamma,\dot{\gamma})\cdot \dot{\gamma}$, we get by
\eqref{eumau13}
\begin{equation}\label{eumau21a}
- \left[D_{\dot{\gamma}} L_h(\gamma(\alpha(t)),\dot{\gamma}(\alpha(t)))\cdot
\dot{\gamma}(\alpha (t))  \partial_u\alpha_\delta |_{u=0}(t)\right]_a^b
= -\left[\left(L_h(\gamma(\alpha(t)),\dot{\gamma}(\alpha(t))) + E\right)\cdot
\partial_u\alpha_\delta |_{u=0}(t)\right]_a^b\; .
\end{equation}
Using \eqref{eumau21} and \eqref{eumau21a}, the boundary terms on the right hand side of
\eqref{eumau10a} and \eqref{eumau11} cancel. Combining \eqref{eumau10}, \eqref{eumau10a} and \eqref{eumau11} yields
\begin{multline}\label{eumau14}
dI|_{(\gamma,\alpha)}\left(\partial_u(\gamma,\alpha)_\delta|_{u=0}\right) \\
=\int_{0}^{\alpha(b)}\left(
 D_{\gamma}L_h(\gamma(s),\dot{\gamma}(s))
 - \frac{d}{ds}
 D_{\dot{\gamma}}L_h(\gamma(s),\dot{\gamma}(s))\right)\cdot
\partial_u \gamma_\delta|_{u=0}(s) \, ds \; .
\end{multline}
For $(\gamma,\alpha) = (\gamma_0,\id)$, the integrand is zero, since the integral curve $(\gamma_0,\dot{\gamma}_0)$
of $X_h$ solves Lagranges equation and thus
\[ dI|_{(\gamma_0,\id)} = 0 \; . \]

{\sl Step 5:}\\
We show (b).

If $\gamma_0$ is a Finslerian geodesic with energy $E$, then $b_{E}(\gamma_0) = (\gamma_0, \id)$ and by
\eqref{eumau40} the integral \eqref{eumau14} is zero for each tangent vector
$\partial_u\gamma_{0,\delta}|_{u=0}$.
By standard arguments it follows that $(\gamma_0,\dot{\gamma}_0)$ solves
Lagranges equation. Thus $\gamma_0$ is a base integral curve of $X_h$.
\qed

\subsection{Application to $H_\ep$}\label{Finseik}

We start with

\begin{proof}[Proof of Lemma \ref{propt}]
(a):\\
These estimates and the regularity follow at once by Hypothesis
\ref{hypdecay}, (a), (i) and (iii).\\
(b):\\
By standard Fourier theory,
$t_0$ is even with respect to $\xi$, i.e. $t_0(x,\xi)=t_0(x,-\xi)$, if and only if for all $\eta\in\Z^d$
\begin{equation}\label{anullgerade}
\tilde{a}_\eta(x) = \tilde{a}_{-\eta}(x)\, .
\end{equation}
To show \eqref{anullgerade} we use that by Hypothesis \ref{hypdecay},(a),(iv) for all $\eta\in\Z^d$,
$x\in\R^d$ and $\ep\in(0,1]$
\begin{equation}\label{asymm}
\tilde{a}_\eta(x) + \ep\,a^{(1)}_{\ep\eta}(x) + R^{(2)}_{\ep\eta}(x,\ep) =
\tilde{a}_{-\eta}(x+\ep\eta)+\ep\,a^{(1)}_{-\ep\eta}(x+\ep\eta) + R^{(2)}_{-\ep\eta}(x+\ep\eta,\ep)\; .
\end{equation}
By Hypothesis \ref{hypdecay},(a),(i) we have 
\begin{equation}\label{tildeastet}
 \tilde{a}_{-\eta} (x + \ep\eta) = \tilde{a}_{-\eta}(x) + O_\eta(\ep) \; .
\end{equation}
Combining \eqref{asymm} and \eqref{tildeastet} leads to
\begin{align}
|\tilde{a}_\eta(x) - \tilde{a}_{-\eta}(x)|  &\leq \ep\, |a^{(1)}_{-\ep\eta}(x+\ep\eta) - a^{(1)}_{\ep\eta}(x) | + 
O_\eta(\ep) + |R^{(2)}_{-\ep\eta}(x+\ep\eta,\ep) - R^{(2)}_{\ep\eta}(x,\ep)| \nonumber \\
&\leq \ep C_{x,\eta} + \ep^2 C_{x,\eta} \label{asymm2}
\end{align}
Since the left hand side of \eqref{asymm2} is independent of $\ep$, it is equal to zero and \eqref{anullgerade} follows.

The analytic continuation of $t_0$ follows at once from Hypothesis \ref{hypdecay},(a),(v), since $a^{(0)}_\gamma(x)$
are the Fourier-coefficients of $t_0(x,\xi)$.\\
(c):\\
By (a), $t_0(x,\xi)=\sum_{\eta\in\Z^d} \tilde{a}_\eta (x) \cos(\eta\cdot\xi)$,
thus its Taylor-expansion at $\xi=0$ yields by Hyp.\ref{hypdecay}(a)(ii)
\begin{equation}\label{entw}
\sum_{\eta\in\Z^d} \tilde{a}_{\eta}(x) \left(1- \frac{1}{2}(\eta\cdot\xi)^2 + O\left(|\xi|^4\right)\right) =
\skp{\xi}{B(x)\xi} + O\left(|\xi|^4\right) \; ,
\end{equation}
where the symmetric $d\times d$-matrix $B$ is given by
\begin{equation}
- \frac{1}{2}\sum_{\eta}\tilde{a}_\eta(x) \eta_\nu\eta_\mu = B_{\nu\mu}(x)
\qquad
\mbox{for}\quad \mu,\nu\in\{1,\ldots,d\}\; , x\in\R^d\label{agamma3}\; .
\end{equation}
Since $\skp{\xi}{B(x)\xi} = -\frac{1}{2\ep^2}\sum_\gamma
a^{(0)}_\gamma (x) (\xi \cdot \gamma)^2$, by Hypothesis \ref{hypdecay},(a),(iii)
and (vi) the matrix $B$ is
positive definite.\\
(d):\\
First we mention that by a short calculation
$\Op_\ep^{\T^d}(e^{-\frac{i}{\ep}\gamma\cdot\xi}) = \tau_\gamma$. This implies
$T_\ep= \sum_\gamma a_\gamma \tau_\gamma = \Op_\ep^{\T^d}(t)$ as operator on $u\in\mathcal{K}(\disk)$
for $t$ given in \eqref{talsexp}.

{\sl Boundedness}:\\
For $u\in \ell^2\left(\disk\right)$, by the Cauchy-Schwarz inequality the $l^2$-norm of
$T_\ep u$ can be estimated as
\begin{align}
\|T_\ep u\|_{\ell^2}^2 &\leq \sum_{x\in\disk}\left(\sum_{\gamma\in\disk} |a_\gamma (x,\ep) u(x+\gamma)|\right)^2 \nonumber\\
& \leq \sum_{x\in\disk}
\left(\sum_{\gamma\in\disk} |a_\gamma (x,\ep)|^2 \left(\frac{|\gamma|}{\ep}\right)^{d+1}\right)^{\frac{1}{2}}
\left(\sum_{\gamma\in\disk}
\left(\frac{|\gamma|}{\ep}\right)^{-(d+1)}|u(x+\gamma)|^2 \right)^{\frac{1}{2}}\; .\label{besch1}
\end{align}
By \eqref{abfallagamma}, the first factor on the right hand side of \eqref{besch1} is bounded uniformly in $x$.
Thus
\[
\|T_\ep u\|^2_{\ell^2}  \leq
C \sum_{\eta\in\Z^d} |\eta|^{-(d+1)} \sum_{x\in\disk} |u(x+\ep\eta)|^2 \leq C \|u\|^2_{\ell^2}\; . \]
Thus $T_\ep$ is a bounded operator on $\ell^2(\disk)$.

{\sl Symmetry:}\\
Since $T_\ep$ is bounded, it is symmetric if and only if for any $x,\gamma\in\disk$
\begin{equation}\label{symmdelta}
\skpd{ T_\ep \delta_{x}}{\delta_{x+ \gamma}} =
\skpd{\delta_x}{T_\ep \delta_{x+\gamma}} \; ,
\end{equation}
where $\delta_x(y):= \delta_{xy}$. Since the left hand side of \eqref{symmdelta} is equal to $a_{-\gamma}(x+\gamma, \ep)$
and the right hand side equals $a_\gamma (x, \ep)$, the statement follows by Hyp.\ref{hypdecay}(a)(iv).

{\sl Boundedness from below}: \\
For $u\in {\mathcal K}(\disk)$, we write
\begin{align}\label{Tbesch}
\skpd{u}{T_\ep u} &= A[u] + B[u]\, \qquad\text{where} \\
A[u] &:= \sum_{x\in\disk}\left\{ a^{(0)}_0(x) |u(x)|^2 + \sum_{\gamma\neq 0} a^{(0)}_\gamma(x) u(x+\gamma)
\bar{u}(x)\right\}\nonumber \\
B[u] &:= \sum_{x,\gamma\in\disk} \left(\ep\, a^{(1)}_\gamma(x) + R^{(2)}_\gamma(x,\ep)\right)\bar{u}(x)u(x+\gamma)
\, .\nonumber 
\end{align}
Then by the exponential decay of $a^{(1)}$ and $R^{(2)}$ with respect to $\gamma$ (Hyp.\ref{hypdecay}(a)(v)) 
\begin{equation}\label{Bu}
|B[u]| \leq \sum_{x,\gamma\in\disk} \left|\ep\, a^{(1)}_\gamma(x) +
R^{(2)}_\gamma(x,\ep)\right|\left(|\bar{u}(x)|^2 + |u(x+\gamma)|^2\right) \leq \ep C\| u\|^2\, . 
\end{equation}
By Hypothesis \ref{hypdecay}(iii) we have
\begin{align}\label{Au}
A[u] &= \sum_{x}\sum_{\gamma\neq 0} a^{(0)}_\gamma (x) \left( u(x+\gamma) \bar{u}(x) - |u(x)|^2\right) \\
&= \frac{1}{2} \left\{ \sum_{\natop{x}{\gamma\neq 0}} a^{(0)}_\gamma (x)
\left( u(x+\gamma) \bar{u}(x) - |u(x)|^2 \right) +
 \sum_{\natop{\tilde{x}}{\tilde{\gamma}\neq 0}} a^{(0)}_{-\tilde{\gamma}} (\tilde{x}+\tilde{\gamma})
 \left( u(\tilde{x}) \bar{u}(\tilde{x} + \tilde{\gamma}) - |u(\tilde{x} + 
 \tilde{\gamma})|^2 \right) \right\}\nonumber \\
  &= - \frac{1}{2}\sum_{\natop{x}{\gamma\neq 0}} a^{(0)}_\gamma (x) \left| u(x) - u(x+\gamma) \right|^2 
   \geq 0 \; , \nonumber
\end{align}
where for the second step we used the symmetry of $T_\ep$ and the substitution $\tilde{x}= x+\gamma$ and 
$\tilde{\gamma}=-\gamma$
and the last estimate follows from Hyp.\ref{hypdecay}(a)(iii). Inserting \eqref{Bu} and \eqref{Au}
in \eqref{Tbesch} gives the stated result.
\end{proof}

Definition \ref{Finslerman} and Theorem \ref{listfinsler}
allow to define a metric adapted to the Hamilton operator $H_\ep$ as follows.\\

\begin{prop}\label{hnullhyper}
The Hamilton function $ \tilde{h}_0: \R^{2d} \cong T^*\R^d \ra \R$ defined in \eqref{tildehnull} is hyperconvex in
each fibre.
\end{prop}

\begin{proof}

We have to show that there exists a constant $\alpha>0$ such that
\begin{equation}\label{htildehyp}
\skp{v}{D^2_\xi \tilde{h}_0(x,\xi)v} \geq\alpha \|v\|^2\quad\text{for all}\quad x,\xi,v\in\R^d\; .
\end{equation}
For simplicity of notation, we will skip the $x$-dependence of $\tilde{h}_0$ and $a^{(0)}_\gamma$.
We have for $\tilde{a}$ defined in \eqref{agammaunep}
\begin{equation}\label{agammapos2}
\skp{v}{D^2_\xi\tilde{h}_0(\xi) v} = -\sum_{\eta\in\Z^d} \tilde{a}_\eta
(\gamma\cdot v)^2 \cosh(\gamma\cdot\xi)
 \; , \qquad \xi,v\in\R^d\; .
\end{equation}
By Hypothesis \ref{hypdecay},(vi), we can choose a basis
$\{\eta^1,\ldots , \eta^d\},\, \eta^j\in\Z^d$ of $\R^d$ with
$\tilde{a}_{\eta^i} < 0$.
Since by Hyp.\ref{hypdecay}(a)(iv)
each summand in \eqref{agammapos2} has positive sign
\begin{equation}\label{hnullhyp1}
\skp{v}{D^2_\xi\tilde{h}_0(\xi) v} \geq  -\sum_{k=1}^d  \tilde{a}_{\eta^k}
(\eta^k \cdot v)^2 \cosh(\eta^k\cdot \xi)
\; , \qquad \xi,v\in\R^d\; .
\end{equation}
We have
$C=\min_k \left(-\tilde{a}_{\eta^k}\right) > 0$, thus \eqref{hnullhyp1} yields
\[ \skp{v}{D^2_\xi\tilde{h}_0(\xi) v} \geq C \sum_{k=1}^d (\eta^k\cdot v)^2
= \skp{v}{Mv} \geq 0\,, \quad\text{for}\quad M= \left(C\sum_k \eta^k_i \eta^k_j\right)\; .\]
The sum can take the value $0$ only if $v=0$ since $\{\eta^k\}$ is a basis of $\R^d$.
Thus $\tilde{h}_0$ is hyperconvex (the lowest eigenvalue of $M$ gives the lower bound for its second derivative).
\end{proof}

Proposition \ref{hnullhyper} leads by Proposition \ref{hyperconvex} and Lemma \ref{propt} to the following corollary.

\begin{cor}\label{hnullhyperregular}
The Hamilton function $ \tilde{h}_0: \R^{2d}\ra \R$ defined in \eqref{tildehnull} is hyperregular and
even and strictly convex in each fibre.
\end{cor}

In the setting of Theorem \ref{listfinsler},
we choose $M=\R^d$, $E=0$ and $h=\tilde{h}_0 = \tilde{t}_0 - V_0$.
Recall that by Hypothesis \ref{hypdecay}, the set $S(0)$ of singular points with respect to the
energy $E=0$ is given by $S(0) = \{0\}$.

\begin{Def}\label{hypohnull}
For the hyperregular Hamilton function $\tilde{h}_0$ given in \eqref{tildehnull}, we define
\begin{equation}\label{elldef}
\ell(x,v) := \begin{cases} \ell_{\tilde{h}_0,0}(x,v) \, ,& x\in \widetilde{M}:= \R^d \setminus \{0\},\, v\in T_x\widetilde{M} \\
0 \,& x=0\; .\end{cases}
\end{equation}
The associated Finsler metric $d_\ell: \R^d\times\R^d \rightarrow [0,\infty)$ is
given by
\begin{equation}\label{Hepmetric}
d_\ell(x_0,x_1) = \inf_{\gamma_{0,1}\in\Gamma(x_0, x_1)} \int_0^{1}
\ell(\gamma(t), \dot{\gamma}(t))\, dt
\, .
\end{equation}
\end{Def}

We notice that it follows from the Definition of $\tilde{\tau}_0$ that $\lim_{x\to 0}\tilde{\tau}_0(x,v) = (0,0)$.
Thus $\ell:\R^{2d}\ra \R$ defined in \eqref{elldef} is continuous.

\subsection{Proof of Theorem 1.5}\label{proofeiko}

In order to prove Theorem \ref{eikonald},
we notice that
if $d_\ell$ is locally Lipschitz continuous, it is
differentiable almost everywhere in both arguments (Rademacher Theorem).\\

{\sl Step 1}: We prove \eqref{d-dgamma}.\\
By the triangle inequality and the definition of $d_\ell(x,y)$, we have for
any $v\in\R^d$ with $|v|=1$ and $\delta > 0$
\begin{equation}\label{nablad1}
d^0(x+\delta v) - d^0(x)  \leq  d_\ell(x,x+\delta v)
\leq \int_0^1 \ell(\gamma_0(t),\dot{\gamma}_0(t))\, dt \; ,
\end{equation}
where $\gamma_0(t) = x + t\delta v$. For this special curve we get by the
homogeneity of $\ell$
\begin{equation}\label{nablad12}
\int_0^1 \ell(\gamma_0(t),\dot{\gamma}_0(t))\, dt \leq \sup_{t\in[0,1]}\ell(x+t\delta v, \delta v)
= \delta \sup_{t\in[0,1]}\ell(x+t\delta v, v)\; ,
\end{equation}
where by a slight abuse of notation $v$ is considered as an element of $T_{x+t\delta v}\R^d$.
Thus \eqref{nablad1} together with \eqref{nablad12} proves \eqref{d-dgamma}.\\

{\sl Step 2}: We prove \eqref{eicungl}.\\
By \eqref{nablad1} and \eqref{nablad12} we have for any $v\in\R^d$ with $|v|=1$ almost everywhere in $x\in\R^d$
\begin{equation}\label{basicnabla}
\nabla d^0(x)\cdot v = \partial_v d^0(x) = \lim_{\delta\to 0} \frac{d^0(x+\delta v) - d^0(x)}{\delta} \leq
\lim_{\delta\to 0} \sup_{t\in[0,1]}\ell(x+ t\delta v,v)  = \ell(x,v)\; .
 \end{equation}
Note that $\nabla d^0 (x)$ can be considered as an element of $T^*_x\R^d$.
Since both sides in \eqref{basicnabla} are positive homogeneous of order one with respect to $v$, we can extend the inequality to all
$v\in\R^d$.
Using \eqref{lalspaar}, the Finsler function $\ell$ can
be written as $\ell(x,v) = \xi_{\tilde{h}_0}(x,\tilde{v})\cdot v$, where $v$ is considered as an element of $T_xM$ and
will be written as $(x,v)$.
It follows from \eqref{basicnabla} that
\begin{equation}\label{xid>0}
\left(\xi_{\tilde{h}_0} (x,\tilde{v}) - \nabla d^0(x)\right)\cdot v \geq 0\; , \qquad (x,v)\in TM\;\; a.e.\,\text{on}\, M \, .
\end{equation}
Since $\tilde{h}_0(x,\xi)$ is differentiable, real valued and convex in each fibre, by \eqref{energie10} the inequality
\[ \tilde{h}_0(x,\xi) \geq \tilde{h}_0(x,\eta) + D_\eta \tilde{h}_0(x,\eta)\cdot\left(\xi-\eta\right) \]
holds for all $x,\xi,\eta \in \R^d$. Thus by setting $\xi = \xi_{\tilde{h}_0}(x,\tilde{v})$ and $\eta = \nabla d^0(x)$,
we get for all $(x,v)\in TM$ the estimate
\begin{equation}\label{eikonalungl1}
\tilde{h}_0(x,\xi_{\tilde{h}_0}(x,\tilde{v})) \geq \tilde{h}_0(x,\nabla d^0(x)) +
D_\xi\tilde{h}_0 (x,\nabla d^0(x))\cdot (\xi_{\tilde{h}_0}(x,\tilde{v})-\nabla d^0(x))\; ,
\end{equation}
where $(x,\tilde{v})\in\mathcal{E}$ is associated to $(x,v)$.
The left hand side of \eqref{eikonalungl1} is by definition of $\tilde{v}$ equal to zero.
Choosing $(x,v)=D_\xi\tilde{h}_0 (x,\nabla d^0(x))$ in equation \eqref{eikonalungl1} yields
\[ 0\geq \tilde{h}_0(x,\nabla d^0(x)) + v \cdot (\xi_{\tilde{h}_0}(x,\tilde{v})-\nabla d^0(x)) \; .\]
Using \eqref{xid>0}, this proves \eqref{eicungl}.\\

{\sl Step 3:} We prove \eqref{eicmitd} (the eikonal equality):\\
We consider the generalized eikonal equation
\begin{equation}\label{eikonal}
\tilde{h}_0(x,\nabla\varphi (x)) = \tilde{t}_0(x,\nabla\varphi (x)) - V_0(x) = 0\; .
\end{equation}

Choose coordinates such that $t_0(x_0,\xi) = |\xi|^2 + O\left(|\xi|^3\right)$ and
$V_0(x) = \sum_{\nu=1}^d \lambda_\nu^2 x_\nu^2 + O\left(|x|^3\right)$.
It is proven in \cite{thesis}
that there exists a unique positive
$\mathscr{C}^{\infty}$-function $\varphi$ defined in a
neighborhood $\Omega$ of 0, solving (\ref{eikonal}) such that
$\varphi$ has an expansion as asymptotic series
\begin{equation}\label{varphi}
\varphi (x) \sim \sum_{\nu =1}^d \frac{\lambda_\nu }{2} x_\nu^2 + \sum_{k\geq 1}\varphi_k(x) \, ,\qquad x\in{\Omega}\, ,
\end{equation}
where each $\varphi_k$ is an homogeneous polynomial of order $k+2$.

In particular, denote by $F_t$ the flow of the Hamiltonian vector field $X_{\tilde{h}_0}$.
Then the Local Stable Manifold Theorem (\cite{abma}) tells us that there is an open neighborhood
$\mathscr{N}$ of $(0,0)$ such that the two submanifolds
\begin{equation}
\Lambda_{\pm}\left(X_{\tilde{h}_0},(0,0)\right) :=
\left\{\left.(x,\xi)\in T^*\R^d \,\right|\, F_t(x,\xi)\to (0,0) \quad\text{for}\quad
t\to \mp\infty\right\}
\end{equation}
exist and are unique in $\mathscr{N}$.
They are called stable ($\Lambda_-$) and unstable ($\Lambda_+$) manifold of $X_{\tilde{h}_0}$ of the
critical point $(0,0)$.
Moreover they are of dimension $d$ and
contained in $\tilde{h}_0^{-1}(0)$. It is shown in \cite{thesis}
that $\Lambda_{\pm}$ are Lagrangian manifolds in $T^*\R^d$
and that the outgoing manifold can be parametrized as $\Lambda_+=\{ (x,\nabla\varphi (x))\,|\, x\in\Omega\}$.
Thus for a given $x\in\Omega$ there exists
an integral curve $\hat{\gamma}_0 := (\gamma_0, \nabla \varphi(\gamma_0))\subset \Lambda_+$ of the Hamiltonian
vector field $X_{\tilde{h}_0}$,
parametrized by $[-\infty,0]$ such that $\hat{\gamma}_0(0) = (x, \nabla\varphi (x))$ and
$\lim_{t \to -\infty} \hat{\gamma}_0(t) = (0,0)$.
Since $\hat{\gamma}_0$ is an integral curve of $X_{\tilde{h}_0}$, it
follows from Hamilton's equations that
\[
(\gamma_0, \dot{\gamma}_0) =
{\mathcal D}_F\tilde{h}_0 \left(\gamma_0, \nabla\varphi (\gamma_0)\right)  \]
and therefore
\[ \left({\mathcal D}_F\tilde{h}_0\right)^{-1}
(\gamma_0, \dot{\gamma}_0) = (\gamma_0, \nabla \varphi (\gamma_0)) \; . \]
Thus
\begin{equation}\label{eikzeig1}
\frac{d}{dt} \varphi\circ \gamma_0 = \nabla \varphi|_{\gamma_0} \cdot \dot{\gamma}_0 = \left({\mathcal D}_F\tilde{h}_0\right)^{-1}
(\gamma_0, \dot{\gamma}_0)\cdot \dot{\gamma}_0\; .
\end{equation}
Since $\hat{\gamma}_0$ is an integral curve, $(\gamma_0(t),\dot{\gamma}_0(t))\in {\mathcal E}$ for all $t$.
Therefore $\tilde{\tau}_0 (\gamma_0,\dot{\gamma}_0)= (\gamma_0,\dot{\gamma}_0)$ and it follows
at once from \eqref{eikzeig1} and the definition of $\ell$
that
\begin{equation}\label{eikzeig2}
\frac{d}{dt} \varphi\circ \gamma_0= \ell(\gamma_0,\dot{\gamma_0})
\end{equation}
The point $x=0$ is a singular point of the Finsler manifold $(\R^d, \ell)$, thus
the base integral curve $\gamma_0:[-\infty,0]\ra \Omega\ni 0$ of $X_{\tilde{h}_0}$
is not a regular curve on a Finsler manifold in the sense of Definition \ref{Finslerman2}.
To avoid this difficulty, we restrict the curve $\gamma_0$ to $[-T,0]$ and set $y_T:= \gamma_0(-T)$.
Then by \eqref{eikzeig2}
\begin{equation}\label{varphialsint}
\varphi(x) - \varphi(y_T)  = \int_{-T}^0 \ell(\gamma_0(t),\dot{\gamma}_0(t))\, dt
\end{equation}
By Proposition \ref{eulermaup} the base integral curve $\gamma_0$ of $X_{\tilde{h}_0}$ is a geodesic with respect to
the associated Finsler function $\ell$.
It is a basic theorem in Finsler Geometry (see Abate-Patrizio \cite{abate}, Theorem 1.6.6),
that geodesics, which are
short enough, actually minimize the curve length among all $\Ce^\infty$-curves (or $\Ce^2$-curves) with the same endpoints.
Thus the length of any short geodesic joining $x$ and $y$ is for $|x-y|$ sufficiently small
equal to the Finsler distance $d_\ell(x,y)$ and
\begin{equation}\label{eikzeig3}
\varphi (x) - \varphi (y_T) = d_\ell(y_T,x)\; .
\end{equation}
Since $y_T\to 0, T\to\infty$ and $d_{\ell}$ and $\varphi$ are continuous in $x=0$, we get
\[ \varphi (x) = d^0(x)\, \quad x\in\Omega \]
for $|x|$ sufficiently small.
\qed

\section{Weighted estimates for Dirichlet eigenfunctions}\label{Kap5a}

\subsection{Preliminary Results}\label{pre}

\begin{Lem}\label{Hepconj}
Assume Hypothesis \ref{hypdecay} and,
for $\Sigma\subset\R^d$, let $H_\ep^\Sigma$ denote the Dirichlet operator introduced in Definition \ref{ell2sigma}.
Let $\varphi:\Sigma \ra \R$ be Lipschitz and constant outside some bounded set.
Then for any real valued $v\in \De (H^\Sigma_\ep)$
\begin{eqnarray*}
 \skpd{\left(e^{\frac{\varphi}{\ep}}H^\Sigma_\ep e^{-\frac{\varphi}{\ep}}\right)v}{v} &=&
\skpd{\left(V_\ep + V_{\ep,\Sigma}^\varphi \right)v}{v}  \\
&&\hspace{0.3cm} - \frac{1}{2}\sum_{x\in \Sigma}\sum_{\gamma \in\Sigma_\ep'(x)}
a_\gamma (x, \ep) \cosh \left(\tfrac{1}{\ep}\varphi (x)-\varphi (x+\gamma)\right) (v(x)-v(x+\gamma))^2 \, ,
\end{eqnarray*}
where $\Sigma_\ep'(x) := \{\gamma\in\disk\,|\, x+\gamma\in\Sigma\}$ and
\begin{equation}\label{Vphiep}
V_{\ep,\Sigma}^\varphi (x) :=
\sum_{\gamma\in\Sigma_\ep'(x)}a_\gamma(x, \ep) \cosh \left(\tfrac{1}{\ep}(\varphi(x)- \varphi(x+\gamma))\right)\; ,
\end{equation}
where the sum on the right hand side converges.
\end{Lem}

\begin{proof}
By use of the symmetry of $T_\ep$ (Lemma \ref{propt}) and since $v$ and $\varphi$ are assumed to be real valued
and $e^{\pm\frac{\varphi}{\ep}} v\in\De (H^\Sigma_\ep)$, we have
\begin{multline*}
\skpd{\left(e^{\frac{\varphi}{\ep}}\id_{\Sigma_\ep}T_\ep \id_{\Sigma_\ep} e^{-\frac{\varphi}{\ep}}\right)v}{v}
= \frac{1}{2}\left[ \skpd{\id_{\Sigma_\ep}T_\ep \id_{\Sigma_\ep}
e^{-\frac{\varphi}{\ep}}\, v}{e^{\frac{\varphi}{\ep}}\,v} +
\skpd{ e^{-\frac{\varphi}{\ep}}\,v}{\id_{\Sigma_\ep}T_\ep \id_{\Sigma_\ep}e^{\frac{\varphi}{\ep}}\,v}\right]\\
= \frac{1}{2}\sum_{x,x+\gamma\in\Sigma} a_\gamma(x, \ep)\left(e^{\frac{1}{\ep}(\varphi (x)-\varphi (x+\gamma))} +
 e^{-\frac{1}{\ep}(\varphi (x)-\varphi (x+\gamma))}\right)v(x+\gamma)v(x) \\
\end{multline*}
Writing $v(x+\gamma)v(x) = v^2(x) - \frac{1}{2} \left( 2 v^2(x) - 2 v(x+\gamma)v(x)\right)$
yields by the definition of $V_{\ep,\Sigma}^\varphi$
\begin{multline}
\skpd{\left(e^{\frac{\varphi}{\ep}}T_\ep e^{-\frac{\varphi}{\ep}}\right)v}{v}  \label{normest1}\\
 = \skpd{V_{\ep,\Sigma}^\varphi v}{v} - \frac{1}{2} \sum_{x,x+\gamma\in\Sigma}
 a_\gamma(x, \ep) \cosh \left(\tfrac{1}{\ep}(\varphi (x)-
\varphi (x+\gamma))\right)(2v^2(x) - 2 v(x)v(x+\gamma)) \; .
\end{multline}
By Hypothesis \ref{hypdecay} we have $a_\gamma (x, \ep) =
a_{-\gamma} (x+\gamma, \ep)$. Thus by use of the substitutions $x'= x+\gamma$ and
$\gamma' = -\gamma$ together with the fact that $\cosh$ is even
\begin{multline}\label{agammatrans}
 \sum_{x,x+\gamma\in\Sigma}a_\gamma(x, \ep) \cosh \left(\tfrac{1}{\ep}(\varphi (x)-
\varphi (x+\gamma))\right) v^2(x) \\
=\sum_{x',x'+\gamma'\in\Sigma}a_{-\gamma'}(x' + \gamma', \ep) \cosh \left(\tfrac{1}{\ep}
(\varphi (x'+\gamma')-\varphi (x'))\right)v^2(x'+\gamma')  \\
= \sum_{x',x'+\gamma'\in\Sigma}a_{\gamma'}(x', \ep) \cosh \left(\tfrac{1}{\ep}(\varphi (x')-
\varphi (x'+\gamma'))\right)v^2(x'+\gamma')
\end{multline}
Inserting \eqref{agammatrans} into \eqref{normest1} gives
\begin{multline*}
\skpd{\left(e^{\frac{\varphi}{\ep}}\id_{\Sigma_\ep}T_\ep\id_{\Sigma_\ep} e^{-\frac{\varphi}{\ep}}\right)v}{v} = \\
\skpd{V_{\ep,\Sigma}^\varphi v}{v} -\frac{1}{2} \sum_{x,x+\gamma\in\Sigma}a_\gamma(x, \ep) \cosh \left(\tfrac{1}{\ep}
(\varphi (x)-
\varphi (x+\gamma))\right)\left(v^2(x) - 2 v(x)v(x+\gamma) + v^2(x+\gamma)\right)
\, .
\end{multline*}
Since $V_\ep$ commutes with $e^{-\frac{\varphi}{\ep}}$, the stated equality follows.
The convergence of the series in \eqref{Vphiep} follows from the decay of $a_\gamma(x,\ep)$ with respect
to $\gamma$ (Hyp.\ref{hypdecay}(v)) together with the assumptions on $\varphi$ and the mean value theorem.
\end{proof}

Lemma \ref{Hepconj} leads to the following norm
estimate, which will be used later on to prove Theorem \ref{weig}.\\

\begin{Lem}\label{HepDchi}
Assume Hypothesis \ref{hypdecay} and,
for $\Sigma\subset\R^d$, let $H_\ep^\Sigma$ denote the Dirichlet operator introduced in Definition \ref{ell2sigma}.
For $E\geq 0$ fixed, let $F_\pm : \Sigma \rightarrow [0,\infty)$ be a pair of functions such that
$F(x) := F_+(x) + F_-(x) > 0 $ and
\begin{equation}\label{F+F-bedingung}
F_+^2(x) - F_-^2(x) = \hat{V}_\ep(x) + V^\varphi_{\ep,\Sigma}(x) - E\; , \qquad x\in\Sigma \; ,
\end{equation}
where $V^\varphi_{\ep,\Sigma}(x)$ is given in \eqref{Vphiep}.
Then
for $v\in \De (H^\Sigma_\ep)$ real-valued with
$Fv\in \ell^2_{\Sigma_\ep}$ and $\varphi:\Sigma \ra \R$ Lipschitz and constant outside
some bounded set, we have for some $C>0$
\begin{equation}\label{FmitF-}
\| Fv\|^2_{\ell^2} \leq 4 \left\| \tfrac{1}{F}\left(e^{\frac{\varphi}{\ep}}(H_\ep^\Sigma - E)
e^{-\frac{\varphi}{\ep}}\right)v\right\|^2_{\ell^2} + 8 \|F_-v\|^2_{\ell^2} + C\ep \|v\|^2 \, .
\end{equation}
\end{Lem}

\begin{proof}
First observe that
\begin{equation}\label{FF-F+}
\| Fv\|^2_{\ell^2} \leq 2\left( \|F_+v\|^2_{\ell^2} + \| F_-v\|^2_{\ell^2} \right)  =
2\left( \|F_+v\|^2_{\ell^2} - \| F_-v\|^2_{\ell^2} \right) + 4 \| F_-v\|^2_{\ell^2}\; .
\end{equation}
By \eqref{F+F-bedingung} one has
\begin{equation}\label{eikoF}
 \| F_+v\|^2_{\ell^2} - \| F_-v\|^2_{\ell^2} =
 \skpd{(\hat{V}_\ep + V_{\ep,\Sigma}^\varphi - E)v}{v} \, .
\end{equation}
Hyp.\ref{hypdecay}(a)(iii) and (v) yields by straightforward calculation
\begin{equation}\label{agamma}
-\frac{1}{2}\sum_{x,x+\gamma\in\Sigma}a_\gamma(x, \ep) \cosh \left(\frac{1}{\ep}(\varphi (x)-
\varphi (x+\gamma))\right) (v(x)-v(x+\gamma))^2\geq -C \ep \|v\|^2\; ,
\end{equation}
since $|\varphi(x+\gamma) - \varphi(x)| \leq |\gamma| \sup_{y\in K} |D\varphi(y)|$ for some compact set $K\subset
\Sigma$. 
Thus it follows from Lemma \ref{Hepconj} and \eqref{agamma} that
\begin{equation}\label{HepundVvarphi}
\skpd{(\hat{V}_\ep + V_{\ep,\Sigma}^\varphi - E)v}{v} - C \ep \|v\|^2 \leq
\skpd{\left( e^{\frac{\varphi}{\ep}}(H^\Sigma_\ep - E)e^{-\frac{\varphi}{\ep}}\right) v}{v}
 \, .
\end{equation}
\eqref{eikoF} and \eqref{HepundVvarphi} yield by use of the Cauchy-Schwarz inequality
\begin{eqnarray}\label{schwarzbinom}
2\left( \|F_+v\|^2_{\ell^2} - \| F_-v\|^2_{\ell^2} \right) - C \ep \|v\|^2 &\leq& 2
\skpd{\left( e^{\frac{\varphi}{\ep}}(H^\Sigma_\ep - E)e^{-\frac{\varphi}{\ep}}\right) v}{v}
 \\
&\leq& 2\sqrt{2}
\left\|\tfrac{1}{F}\left( e^{\frac{\varphi}{\ep}}(H^\Sigma_\ep - E)e^{-\frac{\varphi}{\ep}}\right) v
\right\|_{\ell^2}
\frac{1}{\sqrt{2}}\| Fv\|_{\ell^2} \nonumber\\
&\leq&  2 \left\|\tfrac{1}{F}\left( e^{\frac{\varphi}{\ep}}(H^\Sigma_\ep - E)
e^{-\frac{\varphi}{\ep}}
\right) v\right\|^2_{\ell^2} +\frac{1}{2} \| Fv\|^2_{\ell^2}\, .\nonumber
\end{eqnarray}
Inserting \eqref{schwarzbinom} into \eqref{FF-F+} we get
\[
\| Fv\|^2_{\ell^2} \leq  2 \left\|\tfrac{1}{F}
\left( e^{\frac{\varphi}{\ep}}(H^\Sigma_\ep - E)
e^{-\frac{\varphi}{\ep}}
\right) v\right\|^2_{\ell^2} +\frac{1}{2} \| Fv\|^2_{\ell^2}  +
4 \| F_-v\|^2_{\ell^2} + C\ep \|v\|^2\; . \]
This proves \eqref{FmitF-}.
\end{proof}

\begin{Lem}\label{Philem}
Let $\Sigma\subset \R^d$ be an open bounded region including the point $0$ and such that $d^0\in\Ce^2(\overline{\Sigma})$,
where $d^0(x):= d_\ell(0,x)$ is defined in \eqref{Hepmetric}.
Let $\chi\in{\Ce}^\infty(\R_+,[0,1])$ such that $\chi (r)=0$ for $r\leq
\frac{1}{2}$ and $\chi (r) =1$ for $r\geq 1$. In addition we assume that
$0\leq \chi'(r) \leq \frac{2}{\log 2}$.
For $B>0$ we define
$g: \Sigma \ra [0,1]$ by
\begin{equation}\label{defg}
g(x):= \chi\left(\frac{d^0(x)}{B\ep}\right)\; ,\qquad x\in\Sigma
\end{equation}
and set
\begin{equation}\label{DefPhix}
\Phi (x) := d^0(x) - \frac{B\ep}{2}\log \left(\frac{B}{2}\right) -
g(x)\frac{B\ep}{2} \log \left(\frac{2d^0(x)}{B\ep}\right)\; ,\qquad x\in\Sigma\, .
\end{equation}
Then there exists a constant $C>0$ such that for all $\ep\in(0,\ep_0]$
\begin{equation}\label{deltaphibound}
\left|\partial_\nu\partial_\mu \Phi(x)\right|\leq C \; .
\end{equation}
Furthermore for any $B>0$ there is $C'>0$ such that
\begin{equation}\label{ePhied}
e^{\frac{d^0(x)}{\ep}}\frac{1}{C'}\left(1+\frac{d^0(x)}{\ep}\right)^{-\frac{B}{2}}\leq
e^{\frac{\Phi(x)}{\ep}} \leq
e^{\frac{d^0(x)}{\ep}} C'\left(1+\frac{d^0(x)}{\ep}\right)^{-\frac{B}{2}}\; .
\end{equation}
\end{Lem}

\begin{proof}
We write for simplicity $d(x):=d^0(x)$. First we notice that there exists a $C>0$ such that for $\alpha\in\N^d, |\alpha|\leq 2$
\begin{equation}\label{derivg}
\left| \partial^\alpha g(x) \right| \leq C \ep^{-\frac{\alpha}{2}}\, , \qquad x\in\Sigma \; .
\end{equation}
Here one uses that by \eqref{eicmitd} $d(x) = O(|x|^2)$ and $\nabla d(x) = O(|x|)$ as $|x|\to 0$, thus $x=O(\sqrt{\ep})$
and $\nabla d(x) = O(\sqrt{\ep})$ on $\supp \nabla g \subset \{x\in\Sigma\, |\, \frac{B\ep}{2}\leq d(x) \leq B\ep \}$.
By the definition (\ref{DefPhix}) of $\Phi$ we have
\begin{equation}\label{zweiteabPhi}
\partial_\nu\partial_\mu  \Phi(x) = \partial_\nu\partial_\mu d(x) -
\partial_\nu\partial_\mu\left(g(x) \frac{B\ep}{2}\log \left(\frac{2d(x)}{B\ep}\right)\right)
= A_1 + A_2 + A_3
\end{equation}
with
\begin{align*}
A_1 &:= \partial_\nu\partial_\mu d(x) \\
A_2 &:= -\left\{  \left(\partial_\nu\partial_\mu
g\right)(x) \frac{B\ep}{2}\log
\left(\frac{2d(x)}{B\ep}\right)+ (\partial_\nu g)(x) \frac{B\ep}{2d(x)}
(\partial_\mu d)(x) +
(\partial_\mu g)(x) \frac{B\ep}{2d(x)}(\partial_\nu d)(x) \right\} \\
A_3 &:= g(x) \frac{B\ep}{2d(x)} \left(\frac{(\partial_\nu d)(x)(\partial_\mu d)(x)}{d(x)} +
 (\partial_\nu \partial_\mu d)(x)\right) \; .
\end{align*}
Since $\Sigma$ is bounded, all derivatives of $d$ are at least bounded by a constant
independent of $\ep$, thus $A_1$ is bounded.

Each summand in $A_2$ includes a derivative of $g$ and is therefore supported in the region
$\frac{B\ep}{2}<d(x)<B\ep$. Thus
$1<\frac{2d(x)}{B\ep} < 2$ and from \eqref{eicmitd}, it follows as above
that $\partial_\nu d(x) = O(\sqrt{\ep})$. By \eqref{derivg} $A_2$ is bounded.

To estimate $A_3$, we introduce a constant $\delta>0$ such
that $\{ x\in\Sigma \,|\, d(x)<\delta \}\subset\Omega$ and $\delta\geq\ep_0 B$
and analyze the regions $d(x)<\delta$ and $d(x)\geq \delta$ separately.

{\sl Case 1:} $d(x)<\delta$:\\
By Theorem \ref{eikonald}, we have $\partial_\nu d(x) = O(|x|)$ and $\partial_\nu\partial_\mu d(x) = O(1)$.
Thus there exists a constant $M>0$ such that
\[ \sum_{\nu, \mu} \left|\frac{(\partial_\nu d)(x)(\partial_\mu d)(x)}{d(x)}\right|  +
\left|(\partial_\nu \partial_\mu d)(x)\right| < M \qquad\text{for}
\quad \delta\quad\text{small enough}.\]
Since in addition for $d(x)>\frac{B\ep}{2}$ (on the
support of $g$), the term $\frac{B\ep}{2d}$ is bounded by 1, $A_3$ is bounded
by a constant independent of $\ep$.

{\sl Case 2:} $d(x)\geq \delta$:\\
We use that the derivatives of $d$
are bounded on $\Sigma$ and that $\frac{1}{d(x)}\leq \frac{1}{\delta}$.

Combining Case 1 and 2 we get the boundedness of $A_3$ and thus \eqref{deltaphibound}.\\
To see \eqref{ePhied}, we first note that by definition
\begin{equation}\label{gewichtsfunktion}
e^{\frac{\Phi (x)}{\ep}} = e^{\frac{d(x)}{\ep}}
\left(\frac{B}{2}\right)^{-\frac{B}{2}(g(x)-1)}\left(\frac{d(x)}{\ep}\right)^{-\frac{B}{2}g(x)} \; .
\end{equation}
We notice that for any $y\geq 0$  and for any $B>0$ there exists $\tilde{C}>0$ such that
\[ \frac{1}{\tilde{C}} \leq \frac{y^{\chi(\frac{y}{B})}}{1+y} \leq \tilde{C} \, . \]
Setting $y=\frac{d(x)}{\ep}$, this leads to \eqref{ePhied}.
\end{proof}

\subsection{Proof of Theorem 1.7}\label{proofweig}

We partly follow the ideas in the proof of Proposition 5.5 in Helffer-Sj\"ostrand \cite{hesjo}.

Let
\begin{equation}\label{Diriki}
t_0^\Sigma(x,\xi) :=
\sum_{\gamma\in\Sigma_\ep'(x)}a^{(0)}_\gamma(x)\cos \left(\tfrac{1}{\ep}\gamma\cdot\xi\right)\; , \qquad (x,\xi)\in\Sigma\times\T^d\,
,
\end{equation}
where $\Sigma_\ep'(x) := \{\gamma\in\disk\,|\, x+\gamma\in\Sigma\}$.
We notice that
\begin{equation}\label{tkleiner}
t_0(x,i\xi) \leq t_0^\Sigma(x,i\xi)\; ,
\end{equation}
since $a^{(0)}_\gamma\leq 0$ for $\gamma\neq 0$.
In the following we write for simplicity $d(x):=d^0(x)$.
By Theorem \ref{eikonald}, for any $B>0$ we may choose $\ep_B>0$ such that for all $\ep<\ep_B$
\begin{equation}\label{condB}
V_0(x) + t_0(x,i\nabla d(x)) = 0 \, ,
\qquad x\in\Sigma\cap d^{-1}([0,B\ep)) \, ,
\end{equation}
By \eqref{defg} and \eqref{DefPhix}
\begin{equation}\label{gradPhi}
\nabla \Phi(x) = \nabla d(x) \left\{ 1-\frac{B\ep}{2d(x)}
\chi\left(\frac{d(x)}{B\ep}\right) - \frac{1}{2}
\chi'\left(\frac{d(x)}{B\ep}\right)\,\log
\left(\frac{2d(x)}{B\ep}\right)\right\}\, .
\end{equation}
{\sl Step 1:} We shall show that there is $C_0>0$ independent of $B$ such that
\begin{equation}\label{Vnulltab}
V_0(x) + t_0^\Sigma(x,i\nabla\Phi) \geq \begin{cases} 0\, ,&\, x\in\Sigma\cap d^{-1}([0, B\ep]) \\
\frac{B}{C_0}\ep \, ,&\, x\in\Sigma\cap d^{-1}([B\ep,\infty))
\end{cases}
\end{equation}

{\sl Case 1:} $d(x) \leq \frac{B\ep}{2}$\\
Since $\chi(x)=\chi'(x) =0$ and the eikonal equation
\eqref{eikonal} holds, we get
\begin{equation}\label{wei6}
V_0(x) + t_0(x,i\nabla \Phi(x)) = V_0(x) + t_0(x,i\nabla d(x)) = 0 \, ,\qquad
x\in\Sigma\cap d^{-1}([0, \tfrac{B\ep}{2}]) \; .
\end{equation}
which by \eqref{tkleiner} leads at once to the first estimate in \eqref{Vnulltab} in Case 1.

{\sl Case 2:} $d(x) \geq B\ep$\\
Since $\chi' (x)=0$ in this region, we have by \eqref{gradPhi}
\begin{equation}\label{wei1}
\nabla \Phi (x) = \nabla d(x) \left(1-\frac{B\ep}{2d(x)}\right) \; .
\end{equation}
By Lemma \ref{hnullhyper}, $t_0(x,i\xi)= -\tilde{t}_0(x,\xi)$ is concave with respect to $\xi$, therefore
\begin{equation}\label{convexun}
t_0(x,\lambda i\xi + (1-\lambda)i\eta) \geq \lambda t_0(x,i\xi) +
(1-\lambda)t_0(x,i\eta)\quad\mbox{for}\quad 0\leq \lambda\leq 1,\; \xi,\eta\in \R^d\; .
\end{equation}
We have $0\leq (1-\frac{B\ep}{2d(x)}) \leq 1$.
Thus chosing $\lambda = (1-\frac{B\ep}{2d(x)})$ and $\eta = 0$ in \eqref{convexun} and using
$t_0(x,0)=0$ for all $x\in\disk$, by \eqref{wei1} we get the estimate
\begin{eqnarray}
V_0(x) + t_0(x,i\nabla \Phi (x)) &\geq& V_0(x) + \left(1-\frac{B\ep}{2d(x)}\right)
t_0(x,i\nabla d(x)) \nonumber \\
&\geq& V_0(x)\left( 1-\left(1-\frac{B\ep}{2d(x)}\right)\right) \nonumber\\
&=& V_0(x) \frac{B\ep}{2d(x)} \, ,\label{Vnullcase2}
\end{eqnarray}
where for the second estimate we used that by Theorem
\ref{eikonald} the eikonal inequality $t_0(x,i\nabla d(x)) \geq
-V_0(x)$ holds. It follows from Theorem \ref{eikonald} and Hypothesis \eqref{hypdecay}(b) respectively that
$d(x)=O(|x|^2)$ and $V_0(x) = O(|x|^2)$ for $|x|\to 0$. Since the
region $\Sigma$ was assumed to be bounded, it thus follows that
there exists a constant $C_0>0$ such that
\begin{equation}\label{Vnulld}
C_0^{-1} \leq \frac{V_0(x)}{2d(x)} \leq C_0 \, ,\qquad
x\in\Sigma\cap d^{-1}([ B\ep,\infty)) \; .
\end{equation}
Combining \eqref{tkleiner}, \eqref{Vnullcase2} and \eqref{Vnulld},
we finally get the second estimate in \eqref{Vnulltab}.

{\sl Case 3:} $\frac{B\ep}{2}<d(x)<B\ep$\\
We define
\[ f_1(x) := \frac{B\ep}{2d(x)}\chi\left(\frac{d(x)}{B\ep}\right)\quad\text{and}\quad
f_2(x) := \frac{1}{2}
\chi'\left(\frac{d(x)}{B\ep}\right)\,\log \left(\frac{2d(x)}{B\ep}\right)\; ,\]
such that by \eqref{gradPhi}
\begin{equation}\label{nablaphiundd}
\nabla \Phi(x) = \nabla d(x) (1-f_1(x) - f_2(x))\; .
\end{equation}
Since $1<\frac{2d(x)}{B\ep}<2$, $f_1$ and $f_2$ are non-negative and therefore
$1-f_1(x)-f_2(x)\leq 1$. In addition it follows that $0\leq f_1(x) \leq 1$ and by the assumption
$\chi'(r)\leq \frac{2}{\log 2}$ we get $0\leq f_2(x) \leq 1$. Therefore
$0\leq f_1(x) + f_2(x) \leq 2$ and thus the estimate
\begin{equation}\label{f1f2ab}
|1-f_1(x) - f_2(x)| \leq 1
\end{equation}
holds.
Setting $\lambda(x) :=  1-f_1(x) - f_2(x)$
it follows from \eqref{nablaphiundd} and \eqref{f1f2ab} that
\begin{equation}\label{wei2}
\nabla\Phi(x) = \lambda(x) \nabla d(x)\qquad\text{with}\quad |\lambda(x)|\leq 1\; \quad x\in\R^d\; .
\end{equation}
Thus again from (\ref{convexun}) (with $\eta=0$ and $\xi = \nabla d(x)$) together with \eqref{wei2}, \eqref{tkleiner} and the
fact that $t_0$ is even with respect to $i\xi$
it follows that
\begin{equation}\label{Vnulltlambda}
V_0(x) + t_0^\Sigma(x,i\nabla \Phi (x))  \geq
V_0(x) + |\lambda(x)| t_0(x,i\nabla d(x)) \geq  V_0 (1-|\lambda (x)|)\; ,
\end{equation}
where for the second step we used \eqref{condB}.
Since $|\lambda(x)|\leq 1$ and $V_0\geq 0$, \eqref{Vnulltlambda} gives the first estimate in \eqref{Vnulltab} in Case 3.\\

{\sl Step 2:} We shall show
\begin{equation}\label{VepundVPhi}
\hat{V}_\ep(x) + V^\Phi(x) \geq \begin{cases}  - C_5\,\ep & \qquad\mbox{for}
\quad x\in\Sigma\cap d^{-1}([0, B\ep]) \\
\left(\frac{B}{C_0}-C_5\right)\ep &
\qquad\mbox{for}\quad  x\in\Sigma\cap d^{-1}([B\ep,\infty))\; .
\end{cases}
\end{equation}
for some $C_5>0$ independent of $B$, where $V^\Phi:=V_{\ep,\Sigma}^\Phi$ is defined in \eqref{Vphiep}.\\

We write
\begin{equation}\label{wei5}
\hat{V}_\ep(x) + V^\Phi(x) = \left(\hat{V}_\ep(x) - V_0(x)\right) + \left(V^\Phi(x) - t_0^\Sigma(x,i\nabla\Phi(x))\right)
 + \left(V_0(x) +
t_0^\Sigma(x,i\nabla\Phi(x))\right)
\end{equation}
and give estimates for the differences in the first two brackets on the right hand side.

By Hypothesis \ref{hypdecay} and since $\Sigma$ is bounded, there exists a constant $C_1>0$ such that
\begin{equation}\label{Vephalbbe}
\hat{V}_\ep(x) - V_0(x) \geq - C_1 \ep \, , \qquad x\in\Sigma\; .
\end{equation}
We shall show that
\begin{equation}\label{Vphitsigma}
\left| V^\Phi (x) - t_0^\Sigma(x,i\nabla \Phi(x))\right| \leq \ep C_4 \; .
\end{equation}
Then inserting \eqref{Vphitsigma}, \eqref{Vephalbbe} and \eqref{Vnulltab} in \eqref{wei5} proves \eqref{VepundVPhi}.

Setting (see \eqref{Vphiep})
\[ V^\Phi_0 (x) := \sum_{\gamma\in\Sigma'(x)} a^{(0)}_\gamma (x)
\cosh \left(\tfrac{1}{\ep}(\Phi(x)- \Phi(x+\gamma))\right)\, , \]
we write
\[ V^\Phi(x) - t_0^\Sigma(x,i\nabla\Phi(x))= \left( V^\Phi(x) - V_0^\Phi (x)\right) +
\left( V_0^\Phi (x) - t_0^\Sigma(x,i\nabla\Phi(x)) \right) =: D_1(x) + D_2(x) \]
and analyze the two Summands on the right hand side separately.
Since $\Phi$ is Lipschitz and constant outside of some bounded set, it follows from Hypothesis 
\ref{hypdecay}(a) (as in the proof of \eqref{agamma}) that for some $\tilde{C}>0$
\begin{equation}\label{D1}
\left|D_1(x)\right| = \left|\sum_{\gamma\in\Sigma'(x)}\left(\ep\, a^{(1)}_\gamma(x) + R^{(2)}_\gamma(x,\ep)\right)
\cosh \left(\tfrac{1}{\ep}(\Phi(x)- \Phi(x+\gamma))\right)\right| \leq \tilde{C} \ep\, .
\end{equation}
uniformly with respect to $x$.

We have for $x\in\Sigma$
\begin{equation}\label{Vphiminust}
\left|D_2(x)\right|  \leq \sum_{\gamma\in\Sigma_\ep'(x)}
|a^{(0)}_\gamma(x)| \left|
\cosh\left(\frac{1}{\ep}(\Phi(x) - \Phi(x+\gamma))\right)-
\cosh\left(\frac{1}{\ep}\gamma\nabla\Phi(x)\right)\right|\;.
\end{equation}
By the mean value theorem for $\cosh z$ with  $z_0=\frac{1}{\ep}\gamma\nabla\Phi(x)$ and
$z_1=\frac{1}{\ep}(\Phi(x)-\Phi(x+\gamma))$, we get from $|\sinh x|\leq e^{|x|}$
\begin{multline}\label{mittel}
\left|\cosh\left(\frac{1}{\ep}(\Phi(x) - \Phi(x+\gamma))\right)-\cosh\left(\frac{1}{\ep}\gamma\nabla\Phi(x)\right)\right|\\
\leq \sup_{t\in[0,1]} e^{\left|\frac{1}{\ep}\{(\Phi(x) - \Phi(x+\gamma))t + \gamma\nabla\Phi(x) (1-t)\}\right|}
 \left|\frac{1}{\ep}\{ (\Phi(x)-\Phi(x+\gamma))+\gamma\nabla\Phi(x) \}\right| \; .
\end{multline}
By \eqref{d-dgamma} and the definition \eqref{DefPhix} of $\Phi$ there exist constants
$c_1,c_2>0$ such that
 \begin{equation}\label{lip}
 |\Phi(x)-\Phi(x+\gamma)| \leq c_1|\gamma| \quad\mbox{and}\quad |\gamma\nabla\Phi(x)|\leq
 c_2|\gamma| \; , \qquad x\in\Sigma,\, \gamma\in\Sigma_\ep'(x) \; .
 \end{equation}
\eqref{lip} gives a constant $D>0$
such that
 \begin{equation}\label{estexp}
e^{|\frac{1}{\ep}\{(\Phi(x) - \Phi(x+\gamma))t +
\gamma\nabla\Phi(x) (1-t)\}|}  \leq e^{\frac{D}{\ep}|\gamma|}\; .
\end{equation}
By second order Taylor-expansion
\begin{equation}\label{Phi2Tay}
\frac{1}{\ep} \left| (\Phi(x)-\Phi(x+\gamma))+\gamma\nabla\Phi(x) \right| \leq \sup_{t\in[0,1]}
\frac{1}{\ep} \sum_{\nu,\mu=1}^d \left| \gamma_\nu\gamma_\mu\partial_\nu\partial_\mu\Phi(x+t\gamma)\right|\; .
\end{equation}
Inserting \eqref{deltaphibound} into \eqref{Phi2Tay} shows that there exists a
constant $C_3>0$ independent of the choice of $B$ such that for all $\ep\in(0,\ep_0]$
\begin{equation}\label{dopabab}
\frac{1}{\ep}\left| (\Phi(x)-\Phi(x+\gamma))+\gamma\nabla\Phi(x) \right| \leq
\frac{C_3}{\ep}|\gamma|^2 \; .
\end{equation}
By \eqref{abfallagamma}, inserting \eqref{estexp} and \eqref{dopabab} in \eqref{Vphiminust} we get for any $A>0$ with
$\eta=\frac{\gamma}{\ep}$
\[
\left|V_0^\Phi (x) - t_0^\Sigma(x,-i\nabla\Phi)\right|\leq \sum_{\frac{\eta}{\ep}\in\Sigma_\ep'(x)}
e^{-A|\eta|}e^{D|\eta|}C_3\ep|\eta|^2 \leq \ep
\,\sum_{\eta\in\Z^d}  e^{-|\eta|D'}C_3|\eta|^2 \leq \ep C_4\;,
\]
where $A-D = D'>0$.
This together with \eqref{D1} gives \eqref{Vphitsigma}.\\

{\sl Step 3:} We prove \eqref{weigequ} by use of Lemma \ref{HepDchi}.\\

Choosing $B \geq C_0(1+R_0+C_5)$, we have
\begin{equation}\label{CC0ep}
\left(\frac{B}{C_0}-C_5\right)\ep - E \geq \ep \, ,\qquad E\in[0,\ep R_0]\, .
\end{equation}
Let
\begin{equation}\label{OminusOplus} \O_-:= \{ x\in\Sigma\,|\, \hat{V}_\ep (x) + V^\Phi (x) - E < 0\}\qquad
\text{and}\qquad
\O_+ := \Sigma \setminus \O_- \; ,
\end{equation}
then from \eqref{CC0ep} it follows that $\O_-\subset \{d(x)< \ep B\}$ and
by \eqref{VepundVPhi}
\begin{equation}\label{normVepVphi}
|\hat{V}_\ep (x) + V^\Phi(x) | \leq \ep \,\max \{C_5,R_0\} \,\qquad\text{for all}\quad x\in \O_-\; .
\end{equation}
We define the functions $F_{\pm} : \Sigma \ra [0,\infty)$ by
\begin{equation}\label{F+def}
F_+(x) := \sqrt{\ep\id_{\{d(x)<B\ep \}}(x) + (\hat{V}_\ep(x) + V^\Phi(x) -E)\id_{\O_+}(x)}
\end{equation}
and
\begin{equation}\label{F-def}
F_-(x) := \sqrt{\ep\id_{\{d(x)<B\ep \}}(x) + (E-\hat{V}_\ep(x) - V^\Phi(x))\id_{\O_-}(x)}\; .
\end{equation}
Then $F_\pm$ are well defined and furthermore there exists a constant $C>0$ such that
\begin{equation}\label{Feigenschaft}
F:=F_+ + F_- \geq C \,\sqrt{\ep} >0\, ,\qquad F_-= O(\sqrt{\ep})\qquad \text{and}\qquad F_+^2 - F_-^2 = \hat{V}_\ep +
V^\Phi - E\; .
\end{equation}
Lemma \ref{HepDchi} yields with the choice $v=e^{\frac{\Phi}{\ep}}u$
\begin{equation}\label{HDchiab5}
\left\| F e^{\frac{\Phi}{\ep}}u\right\|^2_{\ell^2} \leq
4 \left\| \tfrac{1}{F}e^{\frac{\Phi}{\ep}}
\left(H_\ep^\Sigma - E\right)u \right\|^2_{\ell^2} +
8 \left\| F_- e^{\frac{\Phi}{\ep}}u\right\|^2_{\ell^2} + C\ep \|u\|^2
\,.
\end{equation}
By \eqref{ePhied} and \eqref{Feigenschaft}
\begin{equation}\label{HDchiab1}
\left\| F e^{\frac{\Phi}{\ep}}u\right\|^2_{\ell^2}\geq C\ep
\left\| \left(1+\tfrac{d}{\ep}\right)^{-\frac{B}{2}}
e^{\frac{d}{\ep}}u\right\|^2_{\ell^2}
\end{equation}
and
\begin{equation}\label{HDchiab2}
\left\| \tfrac{1}{F}e^{\frac{\Phi}{\ep}}
\left(H_\ep^\Sigma - E\right)u \right\|^2_{\ell^2} \leq
C\ep^{-1}\left\| \left(1+\tfrac{d}{\ep}\right)^{-\frac{B}{2}}
e^{\frac{d}{\ep}}\left(H_\ep^\Sigma - E\right)u \right\|^2_{\ell^2}  \, .
\end{equation}
Since $\O_-\subset\{d(x)<B\ep\}$ it follows from the definition of $F_-$ that $\frac{d(x)}{\ep}\leq C$ on the support
of $F_-$.
Therefore by \eqref{ePhied} and \eqref{Feigenschaft} there
exists a constant $C>0$ such that
\begin{equation}\label{HDchiab4}
\left\| F_- e^{\frac{\Phi}{\ep}}u\right\|^2_{\ell^2} \leq
C\ep \left\| u\right\|^2_{\ell^2}\, .
\end{equation}
Inserting (\ref{HDchiab1}), (\ref{HDchiab2}) and
(\ref{HDchiab4}) in
equation (\ref{HDchiab5}) yields with $\tilde{B}:=\frac{B}{2}$
\[
{\tilde C} \ep \left\|
\left(1+\tfrac{d}{\ep}\right)^{-\tilde{B}} e^{\frac{d}{\ep}}u\right\|^2_{\ell^2}
\leq  \ep^{-1}\left\| \left(1+\tfrac{d}{\ep}\right)^{-\tilde{B}}
e^{\frac{d}{\ep}}
\left(H_\ep^\Sigma - E\right)u \right\|^2_{\ell^2} +
\ep \left\| u\right\|^2_{\ell^2}  \,.
\]
This proves \eqref{weigequ}.\\

{\sl Step 4:} We prove \eqref{eigenu}.\\

If $u$ is an eigenfunction of $H_\ep^\Sigma$ with eigenvalue $E$, then the first summand on the right hand side of
\eqref{weigequ} vanishes. The normalization of $u$ leads therefore to \eqref{eigenu}.

\qed

\section{Application to Markov chains}\label{ananwend}

An example of self adjoint difference operators as analyzed above are generators of certain Markov chains and jump
processes. For the sake of the reader we briefly recall some relevant facts on Markov chains.

A Markov chain on $\disk$ is described by means of a
"transition matrix" $P_\ep :\disk\times \disk \rightarrow [0,1]$. $P_\ep$ is a
stochastic matrix, i.e.,
\begin{equation}\label{stochmat}
\sum_{y\in\disk} P_\ep(x,y) = 1\, , \qquad x\in\disk\,.
\end{equation}
We assume that $P_\ep$ satisfies a detailed balance condition, i.e.,
\begin{equation}\label{detailed}
\mu_\ep(x) P_\ep(x,y) = \mu_\ep (y) P_\ep (y,x)
\end{equation}
with respect to some family $\{\mu_\ep\}_{\ep\in (0,\ep_0]}$ of probability measures on $\disk$.
Then $(1-P_\ep)$ defines a self adjoint (diffusion)
operator on $\ell^2\left(\disk,\mu_\ep\right)$ via
\[ (1-P_\ep) u(x) = u(x) - \sum_{y\in\disk} P_\ep (x,y) u(y) \; .  \]
In fact $P_\ep$ is a bounded operator on $\ell^2(\disk, \mu_\ep)$ with $\| P_\ep\| = 1$.
To see this, we first notice that by \eqref{stochmat}
\begin{align*}
|P_\ep u(x)|^2 &\leq \left(\sum_{y\in\disk} P_\ep (x,y)\right)\left(\sum_\gamma P_\ep(x,y) |u(y)|^2\right)\\
&= \sum_\gamma P_\ep(x,y) |u(y)|^2 \; .
\end{align*}
This yields by \eqref{detailed}, the Fubini-Theorem and again \eqref{stochmat}
\begin{align*}
\| P_\ep u\|_{\ell^2(\disk,\mu_\ep)}^2 &= \sum_{x\in\disk} \mu_\ep(x) |P_\ep u(x)|^2 \leq \sum_x \mu_\ep(x)
\sum_y P_\ep(x,y) |u(y)|^2 \\
& = \sum_{y}\left(\sum_x  P_\ep(y, x)\right) \mu_\ep (y) |u(y)|^2  \\
& = \| u\|_{\ell^2(\disk,\mu_\ep)}^2\; ,
\end{align*}
thus $\|P_\ep\| \leq 1$. Since the constant function $u(x)=1$ belongs to $\ell^2(\disk,\mu_\ep)$ and
fulfills $\|P_\ep u\|_{\ell^2(\disk,\mu_\ep)} =  \| u\|_{\ell^2(\disk,\mu_\ep)}$, this proves that $\|P_\ep\| =1$ \\
The symmetry of $P_\ep$ follows from the reversibility condition \eqref{detailed}, since for $u,v\in \ell^2(\disk, \mu_\ep)$
\begin{align*}
\langle u\, , \,P_\ep v\rangle_{\ell^2(\disk,\mu_\ep)} &= \sum_{x\in\disk} \mu_\ep(x) u(x) \sum_{y\in\disk} P_\ep(x,y) v(y) \\
&= \sum_y \sum_x \mu_\ep(y) P_\ep(y,x) u(x) v(y) = \langle P_\ep u\, , \, v\rangle_{\ell^2(\disk,\mu_\ep)}\; .
\end{align*}
Conjugation with respect to the measure $\mu_\ep$ induces a bounded self adjoint
operator $H_\ep := \mu_\ep^{\frac{1}{2}} (1-P_\ep) \mu_\ep^{-\frac{1}{2}}$ on $\ell^2\left(\disk\right)$,
whose restriction to ${\mathcal K}(\disk)$ is given by
\begin{equation}\label{HpropmitP}
H_\ep u(x) = u(x) - \mu_\ep^{\frac{1}{2}}(x)\sum_\gamma P_\ep (x,x+\gamma)\mu_\ep^{-\frac{1}{2}}(x+\gamma)
u(x+\gamma)\; ,\quad u\in {\mathcal K}(\disk)\; .
\end{equation}
Note that ${\mathcal K}(\disk)$ is dense in $\ell^2(\disk)$ and $H_\ep$ is linear continuous and is therefore
completely determined by \eqref{HpropmitP}.

\begin {prop}
The operator $H_\ep := \mu_\ep^{\frac{1}{2}} (1-P_\ep) \mu_\ep^{-\frac{1}{2}}$ on $\ell^2\left(\disk\right)$
is of the form \eqref{Hepein} and fulfills Hypothesis \ref{hypdecay} (a)(iii). If the coefficients $a_\gamma$
have an expansion \eqref{agammaexp}, they also fulfill (ii). 
\end {prop}

\begin{proof}
Setting
\begin{eqnarray}
T_\ep (x)  &:=& \sum_{\gamma\neq 0}\mu_\ep^{\frac{1}{2}}(x) P_\ep(x,x+\gamma)
\mu_\ep^{-\frac{1}{2}}(x+\gamma)(\id - \tau_{\gamma}) \label{Tepprob}\\
V_\ep (x) &:=&\sum_{\gamma\neq 0}\mu_\ep^{\frac{1}{2}}(x)P_\ep(x,x+\gamma)
\left(\mu_\ep^{-\frac{1}{2}}(x)-\mu_\ep^{-\frac{1}{2}}(x+\gamma)\right) \; .\label{Vepprob}
\end{eqnarray}
we have the standard form $H_\ep = T_\ep + V_\ep$, where $V_\ep$ is a potential energy 
(a multiplication operator) and 
$T_\ep$ is of the form described in \eqref{Hepein} with
\begin{eqnarray}
a_0(x,\ep) &=& \sum_{\gamma\neq 0}\mu_\ep^{\frac{1}{2}}(x) P_\ep(x,x+\gamma)
\mu_\ep^{-\frac{1}{2}}(x+\gamma) \geq 0 \label{probanull}\\
a_\gamma(x,\ep) &=& - \mu_\ep^{\frac{1}{2}}(x)P_\ep(x,x+\gamma)\mu_\ep^{-\frac{1}{2}}(x+\gamma)\; , \quad
\gamma\neq 0\; . \label{probagamma}
\end{eqnarray}
Since $P_\ep(x,y)$ and $\mu_\ep(x)$ are non-negative numbers, it follows at once that $a_\gamma(x,\ep)\leq 0$ for all $\gamma\neq 0$ and that $\sum a_\gamma (x,\ep) = 0$. Thus under the assumption \eqref{agammaexp} it follows that Hypothesis \ref{hypdecay} (a)(ii) holds. 
The detailed balance condition for $P_\ep$ ensures the
symmetry of $H_\ep$ and thus of $T_\ep$. By \eqref{symmdelta}, this leads to $a_\gamma (x,\ep) = a_{-\gamma} (x+\gamma,\ep)$ (Hypothesis \ref{hypdecay}(a)(iii)).
\end{proof}

\begin{rem}
The other conditions given in Hypothesis \ref{hypdecay} lead to analog conditions on the transition matrix $P_\ep$
and the reversible measure $\mu_\ep$. For example, condition (a)(iv) on the exponential decay of $a_\gamma$
with respect to $\gamma$, must be reflected by the fact that $P_\ep(x,y)$ is assumed to be exponential small for
$|x-y|$ large. Furthermore in order to fulfill (a)(i), the measure $\mu_\ep$ should be slowly varying. Condition
(a)(v) is a kind of ergodicity condition, which guarantees that jumps in each direction are possible. 
It follows at once from \eqref{Tepprob} and \eqref{Vepprob} that 
for a general probabilistic operator, the potential energy can
be written in terms of the kinetic energy and the measure $\mu_\ep$ as
\begin{equation}\label{VmitTmu}
V_\ep (x) = -  \mu_\ep^{-\frac{1}{2}}(x) \left(T_\ep \mu_\ep^{\frac{1}{2}}\right)(x)\; .
\end{equation}
The assumptions on $V_0$ given in Hypothesis \ref{hypdecay} are conditions on the pair $(\mu_\ep, P_\ep)$.
The class of Markov chains satisfying these conditions is more general than the class of Markov chains analyzed in 
Bovier-Eckhoff-Gayrard-Klein \cite{begk2}, if the Markov chain acts on $\disk$.
\end{rem}

\end{document}